\documentclass[sigconf]{acmart}

\copyrightyear{2025}
\acmYear{2025}
\setcopyright{rightsretained}
\acmConference[PPDP '25]{Proceedings of the 27th International Symposium on Principles and Practice of Declarative Programming}{September 10--11, 2025}{Rende, Italy}
\acmBooktitle{Proceedings of the 27th International Symposium on Principles and Practice of Declarative Programming (PPDP '25), September 10--11, 2025, Rende, Italy}
\acmDOI{10.1145/3756907.3756919}
\acmISBN{979-8-4007-2085-7/25/09}

\newcommand{\seq}[5][\mathcal{N}]{#1;#2;#3;#4\longrightarrow #5}
\newcommand{\logic}{$\mathcal{L}_{LF}$}

\newcommand{\arr}[2]{#1\rightarrow #2}

\newcommand{\ie}{{\it i.e.}}
\newcommand{\eg}{{\it e.g.}}

\newcommand{\etal}{{et al.}}
\newcommand{\rowspace}{\vspace{5pt}}

\newcommand{\app}{\ }

\newcommand{\lflam}[2]{\lambda {#1}.{\mkern 3mu} #2}
\newcommand{\of}{:}
\newcommand{\typeof}[2]{#1\mathrel{:}#2}
\newcommand{\typedpi}[3]{\Pi {#1}{:}{#2}.{\mkern 3mu} #3}
\newcommand{\emptysig}{\cdot}
\newcommand{\emptyctx}{\cdot}

\newcommand{\tm}{\texttt{tm}}
\newcommand{\plussym}{\texttt{plus}}
\newcommand{\plusz}{\texttt{plus-z}}
\newcommand{\pluss}{\texttt{plus-s}}
\newcommand{\plus}[3]{\plussym\ #1\ #2\ #3}
\newcommand{\tmapp}{\texttt{app}}
\newcommand{\tmlam}{\texttt{lam}}
\newcommand{\nat}{\texttt{nat}}
\newcommand{\z}{\texttt{z}}
\newcommand{\s}{\texttt{s}}
\newcommand{\sizesym}{\texttt{size}}
\newcommand{\sizeapp}{\texttt{size-app}}
\newcommand{\sizelam}{\texttt{size-lam}}
\newcommand{\size}[2]{\sizesym\app#1\app#2}
\newcommand{\tpof}[2]{\texttt{type-of}\ #1\ #2}

\newcommand{\lftype}{\mbox{Type}}
\newcommand{\lfwftype}[2]{#1 \vdash_{\Sigma} #2\ \mbox{\tt type}}
\newcommand{\lfwfkind}[2]{#1\ \vdash_{\Sigma} #2\ \mbox{\tt kind}}
\newcommand{\lfwfctx}[1]{\vdash_{\Sigma} #1\ \mbox{\tt ctx}}
\newcommand{\lfwfsig}[1]{\vdash{} #1\ \mathtt{sig}}

\newcommand{\lfprove}[3][\Sigma]{#2 \vdash_{#1} #3}
\newcommand{\lfsynthkind}[4][\Sigma]{\lfprove[#1]{#2}{#3\Rightarrow #4}}
\newcommand{\lfsynthtype}[4][\Sigma]{\lfprove[#1]{#2}{#3\Rightarrow #4}}
\newcommand{\lfchecktype}[4][\Sigma]{\lfprove[#1]{#2}{#3\Leftarrow #4}}

\newcommand{\subord}{\preceq}
\newcommand{\subordinates}[2]{#1 \subord #2}
\newcommand{\tfsubord}[3]{#1 \preceq_{#3} #2}
\newcommand{\nottfsubord}[3]{#1 {\not \preceq}_{#3} #2}

\newcommand{\subordmin}[1]{\vert^{\subord}_{#1}}
\newcommand{\headof}[1]{\vert {#1} \vert}
\newcommand{\ctxmin}[3][a]{#2\subordmin{#1} = #3}

\newcommand{\illformvalidant}[2]{|\!|#2|\!|_{#1} = \mbox{\bf f}}
\newcommand{\illformvalidcons}[2]{|\!|#2|\!|_{#1} = \mbox{\bf t}}
\newcommand{\schemaschemarel}[2]{\sqsupseteq^{#1}_{#2}}
\newcommand{\blockschemarel}{\in}

\newcommand{\prefix}{\trianglelefteq}

\newcommand{\ctxsimfmla}[4][\Gamma]{#2 \gg^{#1}_{#4} #3}
\newcommand{\cesub}[4]{\ctxsimfmla[#1]{#2}{#3}{#4}}

\newcommand{\blockschemarelformula}[4][\Gamma]{#2 \blockschemarel^{#1}_{#3} #4}
\newcommand{\schemarelformula}[4][\Gamma]{#2 \schemaschemarel{#1}{#3} #4}

\newcommand{\prefixof}[3][]{#2 \prefix_{#1} #3}

\newcommand{\blockwftransport}[3][\mathcal{C}]{\prefixof[#1]{#2}{#3}}
\newcommand{\ctxwftransport}[3][\mathcal{C}]{\prefixof[#1]{#2}{#3}}

\newcommand{\subst}[2]{#1[#2]}
\newcommand{\asub}[2]{#1 \llbracket #2 \rrbracket}
\newcommand{\ctxsub}[2]{#1 \lbrack #2 \rbrack}
\newcommand{\asubentry}[3]{\langle {#1}, {#2}, {#3} \rangle}
\newcommand{\actxsub}[3]{\ctxsub{\asub{#1}{#2}}{#3}}
\newcommand{\emptyce}{\cdot}
\newcommand{\erase}[1]{{(#1)}^{-}}
\newcommand{\emptybb}{\cdot}
\newcommand{\declinst}[4]{#2 \leadsto_{\mbox{\sl \scriptsize dec}} #3 \bowtie #4}
\newcommand{\bsinst}[3]{#2 \leadsto_{\mbox{\sl \scriptsize bs}} #3}
\newcommand{\csinst}[3]{#2 \leadsto_{\mbox{\sl \scriptsize cs}} #3}
\newcommand{\csinstone}[3]{#2 \leadsto^1_{\mbox{\sl \scriptsize cs}} #3}
\newcommand{\stlcder}{\vdash_{\mbox{\sl \scriptsize at}}}
\newcommand{\stlctyjudg}[3]{#1 \stlcder #2 : #3}

\newcommand{\STLCGamma}{\Theta}
\newcommand{\initctx}{\Sigma^-}

\newcommand{\acstyping}[1]{\vdash #1\ \mbox{\sl ctx schema}}
\newcommand{\abstyping}[1]{\vdash #1\ \mbox{\sl blk schema}}

\newcommand{\fatm}[3]{\{#1 \vdash \typeof{#2}{#3}\}}
\newcommand{\fconjunct}[2]{#1\wedge #2}
\newcommand{\fdisjunct}[2]{#1\vee #2}
\newcommand{\fimp}[2]{#1\supset #2}
\newcommand{\fall}[2]{\forall #1. #2}
\newcommand{\fexists}[2]{\exists #1. #2}
\newcommand{\genericq}{\mathcal{Q}}
\newcommand{\fgeneric}[2]{\genericq #1. #2}
\newcommand{\fctx}[3]{\Pi\,#1 : #2.#3}

\newcommand{\ftrue}{\top}
\newcommand{\ffalse}{\bot}

\newcommand{\noms}{\mathcal{N}}
\newcommand{\declder}{\vdash_{\mbox{\sl \scriptsize dec}}}
\newcommand{\kindingder}{\vdash_{\mbox{\sl \scriptsize ak}}}
\newcommand{\wfdecls}[3]{#1 \declder #2 \Rightarrow #3}
\newcommand{\wftype}[2]{#1 \kindingder #2\ \mbox{\sl type}}
\newcommand{\emptycs}{\cdot}
\newcommand{\oty}{o}

\newcommand{\permsubs}[2]{\hat{#1}_{#2}}

\newcommand{\apppsubsbd}[3]{\asub{#3}{#1,#2}}
\newcommand{\blkctx}[1]{\mbox{\sl ctx}(#1)}
\newcommand{\ctx}[1]{\mbox{\sl ctx}(#1)}

\DeclareRobustCommand{\Infer}[3][]{%
  \dfrac{\begin{array}{@{}c@{}}\textstyle #3\end{array}}{\textstyle #2}%
  \if\relax\detokenize{#1}\relax\else\ \mathsf{#1}\fi
}

\begin{document}

\title[Transporting Theorems about Typeability in LF Across
      Schematically Defined Contexts]
      {Transporting Theorems about Typeability in LF \\
       Across Schematically Defined Contexts}
\author{Chase Johnson}
\email{joh13266@umn.edu}
\affiliation{%
  \institution{University of Minnesota}
  \city{Minneapolis}
  \state{MN}
  \country{USA}
}
\author{Gopalan Nadathur}
\email{ngopalan@umn.edu}
\affiliation{%
  \institution{University of Minnesota}
  \city{Minneapolis}
  \state{MN}
  \country{USA}
}



\begin{abstract}
  The dependently-typed lambda calculus LF is often used as a vehicle
  for formalizing rule-based descriptions of object systems.
  Proving properties of object systems encoded in this fashion
  requires reasoning about formulas over LF typing judgements.
  An important characteristic of LF is that it supports a
  higher-order abstract syntax representation of binding structure.
  When such an encoding is used, the typing judgements include
  contexts that assign types to bound variables and formulas must
  therefore allow for quantification over contexts.
  The possible instantiations of such quantifiers are usually governed
  by schematic descriptions that must also be made explicit for
  effectiveness in reasoning.
  In practical reasoning tasks, it is often necessary to transport
  theorems involving universal quantification over contexts satisfying
  one schematic description to those satisfying another
  description.
  We provide here a logical justification for this ability.
  Towards this end, we utilize the logic \logic{}, which has
  previously been designed for formalizing properties of LF
  specifications.
  We develop a transportation proof rule and show it to be sound
  relative to the semantics of \logic{}.
  Key to this proof rule is a notion of context schema subsumption that
  uses the subordination relation between types as a means for
  determining the equivalence of contexts relative to individual LF
  typing judgements.
  We discuss the incorporation of this rule into the Adelfa proof
  assistant and its use in actual reasoning examples.
\end{abstract}

\maketitle

\section{Introduction}

The Edinburgh Logical Framework, also known as LF or the
$\lambda\Pi$-calculus, has been used as a vehicle for formalizing
rule-based relational specifications.
When used in this mode, the dependent types of LF provide a natural
means for encoding relations, and terms correspond to derivations for
such relations.
Typing judgements then validate such derivations.

An important feature of LF as a formalization tool is
its support for what is often referred to as \emph{higher-order
abstract syntax}~\cite{pfenning88pldi}.
In this style of representation, binding in the constructs of an
object-language are encoded using abstraction in the meta-language,
which, in this case, is LF.
When such a representation is used, the specification of relations
that involve recursion over the structure of object language
expressions give rise naturally to typing judgements in LF with
non-empty typing contexts.
As a concrete example, consider the size relation between a lambda
term and a natural number.
Under a higher-order abstract syntax encoding in LF, this relation
for an abstraction term will be based on one corresponding to its
body but in a typing context that binds the ``free'' variable in the
body and that simultaneously assigns it the unit size.\footnote{We
will consider this example in more detail in Section~\ref{ssec:lf},
where this intuitive picture will become more precise.}
An important point to note is that the typing contexts that manifest
themselves in typing judgements that arise from such specifications
satisfy a regular pattern.
Thus, in the example considered, they comprise a sequence of a pair of
bindings, each of which identifies a term variable and associates
a size with it.

The broad concern in this paper is with reasoning about specifications
that are encoded in LF in the manner described.
A problem that arises in this context is that of lifting a
property that has been established relative to a typing context that
exhibits one kind of regularity to a situation where it possesses a
different structure.
To illustrate this phenomenon, suppose that we are interested in
showing that every lambda term has an associated size.
Observing that the size of an application term would be obtained by
the addition of the sizes of the two subterms, we see that proving
the property of interest relies crucially on the existence of a
natural number that is the sum of two given natural numbers.
Now, this is a property that is of general interest with regard to
natural numbers and is likely to have been established in a library
that formalizes them.
Since the structure of natural numbers is devoid of binding notions,
the most sensible setting for proving the property would be that where
the typing context is empty.
However, in the example under consideration, we need to use the
property in a situation where the context may be non-empty but has the
structure determined by the definition of the size relation.

It is this kind of lifting of theorems concerning LF
specifications that we refer to as ``transportation,'' and it is this
problem that we study in this paper.
An important part of building such a capability is understanding when
one context can be replaced by another without impacting the validity
of a formula.
A basis for doing so is provided by the ``dependence
relation'' between types that was introduced by Roberto
Virga~\cite{virga99phd} and that is now commonly referred to as the
\emph{subordination} relation~\cite{harper07jfp}.
The negation of this relation allows us to determine when an object of
a particular type may not influence the construction of an object of
another type and, thereby, when a variable of the first type may be
dropped from a context without impacting the validity of the typing
judgement.
However, this idea directly justifies the replacement of one context
by another only in \emph{individual} typing judgements.
To be applicable in our situation, this analysis must be extended to the
setting of properties that relate varied typing judgements.
An additional complexity is that we must consider not just
the replacement of single contexts but, rather, the swapping of
\emph{families of contexts} that are determined by the schematic
descriptions governing the quantification of a context variable.
Our primary endeavor in this work is to describe and
formally justify a mechanism that addresses the two mentioned issues
and that thereby provides a sound basis for theorem
transportation.\footnote{A solution to these issues that
uses subordination more or less directly has been described in a
situation where LF is used also to encode metatheorems about LF
specifications~\cite{harper07jfp}.
This solution relies on the particular, limited form for stating
metatheorems in the underlying approach.
Our goal here is to address the problem in a more general setting.
We discuss this matter again in Section~\ref{sec:conclusion}.
}

To consider the problem in a mathematically precise fashion,
we cast it within a logic for reasoning about LF specifications.
The particular logic we use for this purpose is called
\logic{}~\cite{nadathur22ppdp}.
The atomic formulas in this logic correspond to typing judgements in
LF and they are interpreted by the typing rules of the calculus.
Properties about typing judgements can be expressed using a collection
of logical connectives and quantifiers.
The quantifiers in the logic range over both LF terms and
contexts.
The meaningful description of properties that involve quantification
over contexts requires the domain of such quantifiers to be properly
controlled.
As we have observed in our example, a clue to what the relevant
domain for particular such quantifiers may be is provided by the
structure of the typing judgements involved.
Adapting the idea of regular worlds from the Twelf
system~\cite{Pfenning02guide,schurmann00phd}, \logic{} introduces the
notion of a \emph{context schema} as a type for context variables to
capture this kind of structure.
The question that we are interested in can then  be formalized as the
following: Can we describe circumstances under which the validity of a
formula that involves a universal quantification over a context variable
governed by one context schema implies the validity of the same formula
but with the context variable quantification being governed by a
different context schema?
We answer this question by developing and utilizing the notion of
\emph{context schema subsumption}.

The logic \logic{}, which is initially described via a classical,
subst\-itution-based semantics, has been complemented by a proof
system~\cite{nadathur21arxiv} that has been implemented in the
Adelfa proof assistant~\cite{southern21lfmtp}.
The notion of \emph{context schema subsumption} is brought to fruition
in this setting through the enunciation of a sound proof rule to
support theorem transportation.
We discuss how this proof rule can be incorporated into the Adelfa
system and we illustrate its use in this form through an example.
More details concerning this example and a few others, including
complete proof developments, can be found at the URL
\href{https://adelfa-prover.org/schema-subsumption}{https://adelfa-prover.org/schema-subsumption}.

The rest of this paper is structured as follows.
The next two sections provide an overview of the aspects of LF and
\logic{} that are needed in subsequent discussions.
Section~\ref{sec:subsumption} then develops the idea of context schema
subsumption and establishes its properties.
Section~\ref{sec:rule-and-soundness} uses this idea to describe the
transportation proof rule.
Section~\ref{sec:applications} discusses the integration of the proof
rule into Adelfa and provides some examples of its use.
We conclude the paper with a consideration of related work.

\section{The Edinburgh Logical Framework}\label{sec:lf}

This section has a twofold purpose.
First, it introduces the particular version of LF that will be
used in this paper.
Second, it identifies the relation of subordination between types that
provides the basis for pruning an LF context down to the part that
is relevant to a typing judgement.

\subsection{LF Syntax and Formation Judgements}\label{ssec:lf}

Standard presentations of LF permit lambda terms that are not in
normal form; \eg, see \cite{harper93jacm}.
Such a presentation simplifies the treatment of substitution, but it
complicates arguments concerning adequacy and derivability of typing
judgements that arise when the calculus is used to represent
object systems.
In light of this, we use a version known as \emph{canonical
  LF}~\cite{harper07jfp,watkins03tr} that only admits terms
that are in $\beta$-normal form and where well-typing additionally
requires that they be in $\eta$-long form.
We will refer to this version simply as LF.

\begin{figure}[htpb]
  \[
    \begin{array}{r r c l}
      \mbox{\bf Kinds}                   & K      & ::= & \lftype\ |\ \typedpi{x}{A}{K} \\[5pt]

      \mbox{\bf Canonical Type Families} & A,B    & ::= &
      P\ |\ \typedpi{x}{A}{B}                                                           \\
      \mbox{\bf Atomic Type Families}    & P      & ::= & a\ |\ P\app M                 \\[5pt]
      \mbox{\bf Canonical Terms}         & M,N    & ::= & R\ |\ \lflam{x}{M}            \\
      \mbox{\bf Atomic Terms}            & R      & ::= & c\ |\ x\ |\ R\app M           \\[5pt]
      \mbox{\bf Signatures}              & \Sigma & ::= &
      \emptysig\ |\ \Sigma,c:A\ |\ \Sigma,a:K                                           \\[5pt]
      \mbox{\bf Contexts}                & G      & ::= & \emptyctx\ |\ G,x:A
    \end{array}
  \]
  \caption{The Syntax of LF Expressions}\label{fig:lf-syntax}
  \Description{Described in the text.}
\end{figure}

Expressions in LF belong to three categories: kinds, types, and terms.
Kinds index types and types index terms.
The syntax of these expressions is presented in Figure~\ref{fig:lf-syntax}.
The symbols $c$ and $a$ respectively denote term and type constants
and $x$ denotes a term variable.
Terms and types are divided into canonical forms and atomic forms.
The expressions include two binding forms: $\typedpi{x}{A}{B}$ or
$\typedpi{x}{A}{K}$, which binds the variable $x$ of type $A$ in the
type $B$ or kind $K$, and $\lflam{x}{M}$, which binds $x$ in the term $M$.

\begin{figure*}[hbtp]
  \begin{center}$\displaystyle
      \begin{array}[t]{r l}
        \nat \of   & \lftype               \\
        \z   \of   & \nat                  \\
        \s   \of   & \nat \to \nat         \\\\\\\\
        \tm \of    & \lftype               \\
        \tmapp \of & \tm \to \tm \to \tm   \\
        \tmlam \of & (\tm \to \tm) \to \tm
      \end{array}
      \qquad\qquad
      \begin{array}[t]{r l}
        \plussym \of & \nat \to \nat \to \nat \to \lftype                                                                          \\
        \plusz   \of & \typedpi{N}{\nat}{}\plus{\z}{N}{N}                                                                          \\
        \pluss   \of & \typedpi{N_1}{\nat}{}\typedpi{N_2}{\nat}{}\typedpi{N_3}{\nat}{}                                             \\
                     & \typedpi{D}{\plus{N_1}{N_2}{N_3}}{}                                                                         \\
                     & \plus{(\s\ N_1)}{N_2}{(\s\ N_3)}                                                                            \\\\
        \sizesym\of  & \tm \to \nat \to \lftype                                                                                    \\
        \sizeapp\of  & \typedpi{M_1}{\tm}{}\typedpi{M_2}{\tm}{}\typedpi{N_1}{\nat}{}\typedpi{N_2}{\nat}{}\typedpi{N_3}{\nat}{}     \\
                     & \typedpi{D}{\plus{N_1}{N_2}{N_3}}\size{(\tmapp\ M_1\ M_2)}{(\s\ N_3)}                                       \\
        \sizelam\of  & \typedpi{M}{\tm\to\tm}{}\typedpi{N}{\nat}{}                                                                 \\
                     & \typedpi{D}{\big(\typedpi{x}{\tm}{}\size{x}{(\s\ \z)}\to\size{(M\ x)}{N}\big)}{}\size{(\tmlam\ M)}{(\s\ N)}
      \end{array}$
  \end{center}
  \caption{Encoding Sizes of Untyped Lambda Terms}\label{fig:size-sig}
  \Description{Described in the text.}
\end{figure*}

An important operation on LF expressions that has a bearing even on
typing judgements is that of substitution.
Since the term syntax does not permit $\beta$-redexes, this operation
must normalize terms as it replaces variables.
Care must therefore be exercised to ensure it is a terminating
operation.
Towards this end, substitutions are indexed by arity types that, as we
will see in the next section, can be used to characterize the
functional structure of expressions.
Formally, these are simple types constructed from the atomic type $\oty$
using the function type constructor $\rightarrow$.
A substitution $\theta$ is then a finite set of the form
$\{\asubentry{x_1}{M_1}{\alpha_1}, \ldots,
  \asubentry{x_n}{M_n}{\alpha_n}\}$,
where, for $1 \leq i \leq n$, $x_i$ is a distinct variable, $M_i$ is a
canonical term that is intended to replace the variable, and
$\alpha_i$ is an arity type that governs the replacement.
The application of $\theta$ to an expression $E$ that is a
kind, type, or term corresponds to the replacement of the free
occurrences of the variables $x_i$ by the terms $M_i$ and a subsequent
normalization mediated by the arity type $\alpha_i$;
we refer the reader to~\cite{nadathur21arxiv} for the details of this
operation.
Its application is not always guaranteed to be successful.
However, in the situations that we use it here, it will always yield a
result, which we will denote by $\asub{E}{\theta}$.\footnote{The
notion of substitution we use here is that developed in Section 2.1.2
of \cite{nadathur21arxiv}, which generalizes the hereditary
substitution of \cite{harper07jfp} to the situation where multiple
variables are replaced at the same time.}

There are seven typing judgements in (canonical) LF:
$\lfwfsig{\Sigma}$
that ensures that the constants declared in a signature are distinct
and their type or kind classifiers are well-formed; $\lfwfctx{\Gamma}$
that ensures that the variables declared in a context are distinct and
their type classifiers are well-formed in the preceding declarations
and well-formed signature $\Sigma$; $\lfwfkind{\Gamma}{K}$ that
determines that a kind $K$ is well-formed with respect to a
well-formed signature and context pair; $\lfwftype{\Gamma}{A}$ and
$\lfsynthkind{\Gamma}{P}{K}$ that check, respectively, the formation
of a canonical and atomic type relative to a well-formed
signature, context and, in the latter case, kind; and
$\lfchecktype{\Gamma}{M}{A}$ and
$\lfsynthtype{\Gamma}{R}{A}$ that ensure, respectively, that a
canonical and atomic term are well-formed with respect to a
well-formed signature, context and canonical type.
We refer the reader to \cite{harper07jfp} or \cite{nadathur22ppdp}
for the rules defining these judgements.
We limit ourselves to two remarks here.
First, the formation rules for canonical types and terms, \ie, the
rules for $\lfsynthkind{\Gamma}{P}{K}$ and
$\lfsynthtype{\Gamma}{R}{A}$ will require the substitution of the
argument term in an LF expression.
This substitution will be indexed by the arity type that is obtained
from the one determined for the argument term through an
``erasure'' operation denoted by $\erase{\cdot}$ and defined as
follows: $\erase{P} = \oty$ and $\erase{\typedpi{x}{A_1}{A_2}} =
  \arr{\erase{A_1}}{\erase{A_2}}$.
Second, all the judgements forms other than $\lfwfsig{\Sigma}$ are
parameterized by a signature that remains unchanged in the course of
their derivation.
In the rest of this paper we will assume a fixed signature that has
been verified to be well-formed at the outset.

In the typical use of LF, an object system is specified through a
signature.
One set of declarations in the signature identify types and
constructors for representing the objects of interest in the system.
Further declarations identify (dependent) types that encode relations
over these objects.
Finally, these types are used to identify constants that encode rules
for deriving the corresponding relations.
Figure~\ref{fig:size-sig}, which presents a set of declarations that
encode the natural numbers and then uses this to describe the size
relation relative to the terms in the untyped lambda calculus,
exemplifies this paradigm.
The encoding of natural numbers and the rules for the addition
relation over such numbers, encoded here by the type constructor
\plussym, follows the expected lines.
The representation of lambda terms---in particular, the representation
of abstractions---makes use of the higher-order abstract syntax
approach.
Thus, an object language term of the form $\lflam{x}{M}$ will be
represented by the LF expression $(\tmlam\app
(\lflam{x}{\overline{M}}))$, where $\overline{M}$ is a representation
of $M$ and the latter abstraction is one in the meta-language, \ie, in
LF.
The encoding of the rule that specifies the size relation for
abstractions in the object language is also interesting to note.
This encoding is realized through the LF constant $\sizelam$.
Using it to construct a term that has the type
$(\size{(\tmlam\ M)}{N})$ relative to an LF typing context $\Gamma$
will require us to construct a term that has the type
$(\size{(M\app x)}{N'})$ relative to a context that enhances $\Gamma$
with the bindings $x\of \tm, y : \size{x}{(\s\app z)}$; if this
construction succeeds and if $N$ can be matched with $(\s\app N')$,
then the overall construction will succeed.
The derivation of the size relation for a closed lambda term
represented by $M$ would correspond under this encoding to showing the
inhabitation of the type $(\size{M}{N})$ for some $N$ relative to the
empty typing context.
It is easy to see now that the typing contexts that arise in the
course of this task will all comprise of repetitions of ``blocks'' of
bindings of the form $x\of \tm, y : \size{x}{(\s\app z)}$.
It is this kind of regularity in the shape of contexts that will need
to be captured in a logic for reasoning about LF specifications.

\subsection{Subordination and Context Redundancies}\label{ssec:subord}

Not all the associations contained in a context may be relevant to a
typing judgement.
An analytical approach to determining which bindings are irrelevant
and hence may be pruned from the context would be useful in
determining context equivalence: if two contexts can be pruned
to the same ``core,'' then the validity of a typing judgement is
invariant under a replacement of one by the other.
The subordination relation between types that was identified in
\cite{virga99phd} provides a means for such pruning.
Definition~\ref{def:subordination} presents this relation.
Use is made here of the operation $\headof{\cdot}$ that identifies the
``head constant'' of a type and that is given as follows: $\headof{(a
    \app M_1\app \ldots\app M_n)}$ is $a$ and
$\headof{\typedpi{x}{A}{B}}$ is $\headof{B}$.

\begin{definition}
  \label{def:subordination}
  The subordination relation induced by a well-formed signature
  $\Sigma$ is the smallest relation between the type constants
  identified by $\Sigma$ that satisfies the following conditions:

  \begin{enumerate}
    \item Index subordination: For all declarations in $\Sigma$ of the form
          \begin{enumerate}
            \item $a \colon \typedpi{x_1}{A_1} \ldots
                    \typedpi{x_n}{A_n}{\lftype}$, it is the case that
                  $\headof{A_i} \subord a$; and
            \item $c \colon \typedpi{x_1}{A_1} \ldots \typedpi{x_n}{A_n}{A}$, it
                  is the case that $\headof{A_i} \subord \headof{A}$
          \end{enumerate}
          for $1 \leq i \leq n$.
    \item Reflexivity: For all $a$ declared in $\Sigma$, $a \subord a$.
    \item Transitivity: If $a_1 \subord a_2$ and $a_2 \subord a_3$ then
          $a_1 \subord a_3$.
  \end{enumerate}
\end{definition}

As an illustration, the subordination relation induced by the
signature in Figure~\ref{fig:size-sig} is the following:
\[
  \begin{array}{c@{\hspace{20pt}}c@{\hspace{20pt}}c}
    \subordinates{\tm}{\tm}           & \subordinates{\tm}{\size}         & \subordinates{\nat}{\nat}         \\
    \subordinates{\nat}{\plussym}     & \subordinates{\nat}{\size}        & \subordinates{\plussym}{\plussym} \\
    \subordinates{\plussym}{\sizesym} & \subordinates{\sizesym}{\sizesym}
  \end{array}
\]
We write $a \not\subord b$ to denote the fact that $a \subord b$ does
not hold.
Subordination is extended from a relation between type constants to
one between types by defining $A \subord B$ to
hold exactly when $\headof{A} \subord \headof{B}$ does.

The intuitive understanding of the subordination relation is that a
term of type $B$ can appear in a term of type $A$ only if $B \subord
  A$.
Thus, we may drop assignments of type $B$ from a context without
impacting the assessment that a term $M$ has type $A$ if $B
  \not\subord A$.
This idea underlies the minimization operation on contexts given by
the rules below:
\begin{center}
  \[
    \Infer{\ctxmin[A]{\emptyce}{\emptyce}}{}
    \qquad
    \Infer{\ctxmin[A]{(\Gamma, x: B)}{(\Gamma', x: B)}}{\ctxmin[A]{\Gamma}{\Gamma'} \qquad B \subord A}
    \qquad
    \Infer{\ctxmin[A]{(\Gamma, x:B)}{\Gamma'}}{\ctxmin[A]{\Gamma}{\Gamma'} \qquad B \not\subord A}
  \]
\end{center}
Although defined as a relation, for any given context $\Gamma$ and type
$A$, there is exactly one $\Gamma'$ such that $\ctxmin[A]{\Gamma}{\Gamma'}$ holds.
By an abuse of notation, we will write $\Gamma \subordmin{A}$ to
denote that $\Gamma'$.

The properties concerning context minimization that are need in this
paper are contained in the following proposition that follows
from the results in~\cite{harper07jfp}.

\begin{proposition}
  \label{thm:trans-canon-form}
  Let $\Gamma$ and $\Gamma'$ be contexts and let $A$ be a type such that
  $\lfwfctx{\Gamma}$ and $\Gamma\subordmin{A} = \Gamma'$. Then the following
  are true:
  \begin{enumerate}
    \item $\lfwfctx{\Gamma'}$ holds.

    \item $\lfwftype{\Gamma}{A}$ holds if and only if $\lfwftype{\Gamma'}{A}$
          holds.

    \item  If $\lfwftype{\Gamma}{A}$ is derivable, then
          $\lfchecktype{\Gamma}{M}{A}$ holds if and only if
          $\lfchecktype{\Gamma'}{M}{A}$ does. This equivalence extends to the
          judgements $\lfsynthtype{\Gamma}{R}{A}$ and
          $\lfsynthtype{\Gamma'}{R}{A}$, as well as $\lfsynthkind{\Gamma}{P}{K}$
          and $\lfsynthkind{\Gamma'}{P}{K}$.
  \end{enumerate}
\end{proposition}

\section{Reasoning About LF Judgements}\label{sec:lf-logic}

We are often interested in formalizing properties of LF specifications
since these reflect properties of the object systems they describe.
The logic \logic{} provides a means for doing this.
We describe this logic below towards providing a basis
for posing and addressing the primary issue of interest in this paper.

\subsection{Formulas in the Logic and their Meaning}

The logic \logic{} is parameterized by an LF signature that we will
denote by $\Sigma$.
The atomic formulas of \logic{} are, at the outset, expressions that
encode LF typing judgements of the form
$\lfchecktype{\Gamma}{M}{A}$.\footnote{With one difference in
interpretation: as will become clear in due course, the atomic
formulas assert the well-formedness of $\Gamma$ and $A$ rather than
assuming them. We also note that the
primary interest within LF is in typing judgements of the form
$\lfchecktype{\Gamma}{M}{A}$; judgements of the form
$\lfsynthtype{\Gamma}{R}{P}$ are useful mainly in defining those of
the former kind for atomic types. In reasoning about typing
judgements, it is therefore possible to dispense with the latter form
of judgement via a derived rule for the judgement
$\lfchecktype{\Gamma}{M}{A}$ when $A$ is
atomic~\cite{nadathur21arxiv}.}
However, the syntax of these formulas differs somewhat
from that of the LF expressions.
To begin with, term variables that are bound by quantifiers
are permitted to appear in types and terms.
These variables have a different logical character from the variables
that are bound by an LF context: they may be instantiated by terms in the
domains of the quantifiers, whereas the variables bound by declarations
in an LF context represent fixed entities that are also distinct from
all other similar entities within the typing judgement.
To capture the role of the variables bound in an LF
context, they are represented by \emph{nominal
  constants}~\cite{gacek11ic,tiu06lfmtp}; these are entities that we
represent by the symbol $n$ possibly with subscripts and
that behave like constants except that they may be permuted in
atomic formulas without changing the logical content.
To support this treatment, nominal constants are also allowed to
appear in expressions in the logic that encode LF types and terms.
Finally, contexts may be represented by variables that can be
quantified over.
More specifically, the atomic formulas in \logic{} take the form
$\fatm{G}{M}{A}$, where $M$ and $A$ represent an LF term and type with
the caveats just described and $G$ constitutes a \emph{context
expression} whose syntax is given by the following rule:
\[
  G \quad ::= \quad \Gamma\ |\ \cdot\ |\ G,n:A
\]
The symbol $\Gamma$ here denotes the category of variables that range
over contexts.

There is actually another structural requirement that is imposed on
the terms and types that appear in the formulas in \logic{}.
They are expected to be in canonical form and to respect the functional
structure determined by the dependent types and kinds associated with
the constants and variables appearing in them.
This requirement is realized through an \emph{arity typing} relation
$\stlctyjudg{\STLCGamma}{M}{\alpha}$ and an \emph{arity kinding}
relation $\wftype{\STLCGamma}{A}$, in which $\STLCGamma$ is an
\emph{arity context} that assigns arity types to (term) constants, nominal
constants and variables, $\alpha$ is an arity type, and $M$ and $A$
are, respectively, an LF term and type in which nominal constants may
appear.
We elide the specific definitions of these relations---the details may
be found in~\cite{nadathur21arxiv}.
The type assignment for constants in the arity context is obtained by
taking each assignment of the form $c:A$ in $\Sigma$ and replacing it
with $c:\erase{A}$; we denote this context induced by $\Sigma$ by
$\initctx$.
For nominal constants, we assume these are drawn from the set
$\mathcal{N}$ that automatically assigns an arity type to each and, in
fact, provides a denumerably infinite number of nominal constants for
each arity type.
The arity types for variables is determined by the location of the
term or type in an expression in the logic.

\begin{figure*}[tbhp]
  \centering

  \begin{tabular}{cc}

    $ \Infer
        { \wfdecls{\STLCGamma}{\emptybb}{\STLCGamma} }
        { } $

    &

    $ \Infer
        { \wfdecls{\STLCGamma}{\Delta,\, y\!:\!A}
                 {\STLCGamma' \cup \{\,y:\erase{A}\,\}} }
        { \wfdecls{\STLCGamma}{\Delta}{\STLCGamma'} \qquad
          \text{$y$ is not assigned by }\STLCGamma' \qquad
          \wftype{\STLCGamma'}{A} } $

  \end{tabular}

  \rowspace{}\rowspace{}
  
  \begin{tabular}{c}

    $ \Infer
        { \abstyping{\{x_1:\alpha_1,\ldots, x_n:\alpha_n\}\,\Delta} }
        { \text{$x_1,\ldots,x_n$ are distinct variables} \qquad
          \wfdecls{\initctx \cup \{\,x_1:\alpha_1,\ldots,x_n:\alpha_n\,\}}
                  {\Delta}
                  {\STLCGamma'} } $

  \end{tabular}

  \rowspace{}\rowspace{}

  \begin{tabular}{cc}

    $\Infer { \acstyping{\emptycs} } { } $

    &

    $\Infer
        { \acstyping{\mathcal{C},\,\mathcal{B}} }
        { \acstyping{\mathcal{C}} \quad \abstyping{\mathcal{B}} } $

  \end{tabular}

  \caption{Well-formedness Judgements for Block and Context Schemas}
  \label{fig:schematyping}
  \Description{Described in the text.}
\end{figure*}

\begin{figure*}[tbhp]
  \centering

  \begin{tabular}{cc}

    $ \Infer
        { \declinst{}{\emptybb}{\emptyce}{\emptyset} }
        { } $

    \quad & \quad

    $ \Infer
        { \declinst{}{\Delta,\,y\!:\!A}{G,\,n\!:\!\asub{A}{\theta}}
                 {\theta \cup \{ \langle y,n,\erase{A} \rangle \}} }
        { \declinst{}{\Delta}{G}{\theta} \qquad
          n\!:\!\erase{A} \in \mathcal{N} } $

  \end{tabular}

  \rowspace{}\rowspace{}

  \begin{tabular}{c}

    $ \Infer
        { \bsinst{}
            {\{x_1\!:\!\alpha_1,\ldots,x_n\!:\!\alpha_n\}\,\Delta}
            {\asub{G'}{\{\langle x_i,t_i,\alpha_i\rangle \ \vert\ 1 \le i \le n \}} } }
        { \declinst{}{\Delta}{G'}{\theta} \qquad
          \{\, \stlctyjudg{{\mathcal{N}} \cup \initctx}{t_i}{\alpha_i}
              \ \vert\ 1 \le i \le n \,\} } $

  \end{tabular}

  \rowspace{}\rowspace{}

  \begin{tabular}{cccc}

    $ \Infer
        { \csinstone{}{\mathcal{C},\mathcal{B}}{G} }
        { \bsinst{}{\mathcal{B}}{G} } $

    \quad & \quad

    $ \Infer
        { \csinstone{}{\mathcal{C},\mathcal{B}}{G} }
        { \csinstone{}{\mathcal{C}}{G} } $

    \quad & \quad

    $ \Infer
        { \csinst{}{\mathcal{C}}{\emptyce} }
        { } $

    \quad & \quad

    $ \Infer
        { \csinst{}{\mathcal{C}}{G,\,G'} }
        { \csinst{}{\mathcal{C}}{G} \qquad \csinstone{}{\mathcal{C}}{G'} } $

  \end{tabular}

  \caption{Instantiating a Context Schema}
  \label{fig:ctx-inst}
  \Description{Described in the text.}
\end{figure*}

As we have already noted, in typical typing scenarios, instantiations
for context variables adhere to a regular structure.
To be able to reason effectively about such judgements, the logic must
provide a means for imposing such structural constraints on
instantiations.
This requirement is realized by associating a special kind of typing
with context variables that is inspired by the notion
of \emph{regular worlds} in
Twelf~\cite{Pfenning02guide,schurmann00phd}.
Formally, it is based on the use of \emph{context schemas} whose forms
are given by the following rules:
\[\begin{array}{rrcl}
    \mbox{\bf Block Declarations} & \Delta      & ::= & \emptybb\ \vert\ \Delta, y : A              \\
    \mbox{\bf Block Schema}       & \mathcal{B} & ::= & \{x_1:\alpha_1,\ldots, x_n:\alpha_n\}\Delta \\
    \mbox{\bf Context Schema}     & \mathcal{C} & ::= & \cdot\ \vert\ \mathcal{C}, \mathcal{B}
  \end{array}\]
Conceptually, a context schema comprises a collection of block
schemas, each of which is parameterized by a set of variables and
assigns types that are well-formed in the arity kinding sense to
distinct variables that are also distinct from the variables
parameterizing the block schema.
These requirements are realized by the judgement
$\acstyping{\mathcal{C}}$ that is defined by the rules in
Figure~\ref{fig:schematyping}.
As an example, consider the context schema $\mathcal{B}_1,
\mathcal{B}_2$ where $\mathcal{B}_1$ is the block schema $\{\}(x_1 :
\tm, y_1 : \size{x_1}{(\s\ \z)})$ and $\mathcal{B}_2$ is the block
schema $\{T:o\}(x_2 : \tm, y_2: \tpof{x_2}{T})$; we will refer to this
context schema as $\mathcal{C}$.
It is easy to see that this is a well-formed context schema, \ie, that
$\acstyping{\mathcal{C}}$ is derivable.
Context schemas are intended to represent context expressions
that are obtained via repeated instantiations of their block schemas.
In this sense, the block schema $\mathcal{B}_1$ provides a template for
generating a context expression that might arise when we try to
check an LF term that is intended to encode a derivation that another
LF term encoding an (object-language) lambda term has a size.
Similarly, $\mathcal{B}_2$ provides a template for generating context
expressions that might arise when we try to check an LF term that is
intended to represent a typing derivation for an (object-language)
lambda term.
The context schema $\mathcal{C}$ allows both kinds of context
expressions to be generated as well as those that mix the two kinds of
``blocks'' of type assignments.

The intended meaning of context schemas is codified in the judgement
$\csinst{}{\mathcal{C}}{G}$ defined in Figure~\ref{fig:ctx-inst},
which identifies $G$ as an instance of the context schema
$\mathcal{C}$.
This definition uses the relation
$\bsinst{}{\mathcal{B}}{G}$, also formalized in the same figure, that
identifies the context expression $G$ as
an instance of a block schema $\mathcal{B}$ if $G$ is obtained by
instantiating the variables parameterizing $\mathcal{B}$ with (closed)
well-formed terms of the right arity types and replacing the variables
it assigns types to with nominal constants.
In elaborating the latter aspect, use is made of the further relation
$\declinst{}{\Delta}{G}{\theta}$, which holds if $G$ is a context
expression obtained from the block declaration $\Delta$ by replacing
the variables that it assigns types to with nominal constants; the
replacement is recorded in the substitution $\theta$.
The operation $\asub{\cdot}{\cdot}$ that is employed in these rules
represents substitution application.
The application of a substitution to \logic\ terms and types is
identical to that for LF, with the observation that nominal constants
are treated just like other constants.
The application to context expressions leaves context variables
unaffected and simply distributes to the types in the explicit
bindings.\footnote{In the use manifest in Figure~\ref{fig:ctx-inst},
context expressions do not have context variables in them, but they
may have such variables in later uses of the substitution operation.}
As illustrations of this definition, we see that our example context
schema $\mathcal{C}$ has as instances the context expressions
$(n_1 : \tm, n_2: \size{n_1}{(\s\ \z)})$ (resulting from instantiating
the block schema $\mathcal{B}_1$), $(n_1 :
\tm, n_2: \tpof{n_1}{T})$ (resulting from instantiating the block schema
$\mathcal{B}_2$), and $(n_1 : \tm, n_2: \size{n_1}{(\s\ \z)}, n_3 :
\tm, n_4 : \tpof{n_3}{T})$ (resulting from instantiating both
$\mathcal{B}_1$ and $\mathcal{B}_2$ once), where $T$ is an LF
expression representing a type in the corresponding signature.

The formulas of \logic\ are given by the following syntax rule:
\[\begin{array}{lrl}
    \mbox{\bf Formulas} & F ::= &
    \fatm{G}{M}{A}\ |\ \ftrue\ |\ \ffalse\ |\ \fimp{F_1}{F_2}\ |\ \fconjunct{F_1}{F_2}\ |                                                 \\
                        &       & \fdisjunct{F_1}{F_2}\ |\ \fall{x:\alpha}{F}\ |\ \fexists{x:\alpha}{F}\ |\ \fctx{\Gamma}{\mathcal{C}}{F} \\
  \end{array}
\]
The symbol $\Pi$ represents universal quantification pertaining to
contexts; note that such quantification is qualified by a context
schema.
The symbol $x$ represents a term variable, i.e., the logic permits
universal and existential quantification over LF terms.
The symbol $\alpha$ that annotates such variables represents an arity
type.
A formula is said to be well-formed relative to an assignment $\Psi$
of arity types to term variables and an assignment $\Xi$ of context
schemas to context variables if the following conditions hold: all
context variables that occur in the formula do so within the scope of
a context-level quantifier or are assigned a type by $\Xi$, and the
terms and types in the formula are well-formed in an arity typing
sense with respect to an arity context given by $\initctx$,
$\mathcal{N}$, and an arity typing for variables determined first by
the (term-level) quantifiers within whose scope they appear and then
by $\Psi$.
A formal description of these requirements may be found
in~\cite{nadathur21arxiv}.
In what follows, we shall assume that any formula of interest is
well-formed with respect to some $\Psi$ and $\Xi$.

The meaning of a formula is clarified by instantiating the quantifiers
and free variables in it with closed expressions, \ie, expressions
devoid of term and context variables, and then ascertaining the
validity of the result using the LF typing rules.
The instantiation must respect typing constraints.
For a term variable quantifier, this requirement translates into the
substitution term being well-formed in an arity typing sense and the
substitution being arity type preserving.
For a context variable quantifier, the instantiation must respect the
governing context schema: if $\Gamma$ is a context variable that is to
be substituted for by $G$ and $\Gamma$ is qualified by the context
schema $\mathcal{C}$, then the judgement $\csinst{}{\mathcal{C}}{G}$
must be derivable.

Quantifier instantiation requires the notion of substitution into a
well-formed formula.
A context variable substitution $\sigma$ has the form
$\{G_1/\Gamma_1,\allowbreak\ldots,\allowbreak G_n/\Gamma_n\}$ where,
for $1 \leq i \leq n$, $\Gamma_i$ is a context variable and $G_i$ is a
context expression.
The application of $\sigma$ to a formula $F$, denoted by
$\subst{F}{\sigma}$, corresponds to the replacement of the free
occurrences of the variables $\Gamma_1,\ldots,\Gamma_n$ in $F$ by the
corresponding context expressions, renaming bound context variables
appearing in $F$ away from those appearing in $G_1,\ldots,G_n$.
For term variables, we adapt the substitution operation described for
LF expressions to formulas in \logic.
This operation distributes over quantifiers and logical symbols in
formulas, respecting the scopes of quantifiers through the necessary
renaming.
The application to an atomic formula of the form $\fatm{G}{M}{A}$
distributes to the component parts.

Definition~\ref{def:validity} formalizes validity for \logic{} formulas.
The first clause in this definition uses the judgements
$\lfwfctx{G}$, $\lfwftype{G}{A}$, and $\lfchecktype{G}{M}{A}$, where
$G$, $M$, and $A$ are closed expressions.
Thus, these expressions are identical to the ones in LF with the
exception that what were earlier understood to be variables bound in a
context are now denoted by nominal constants.
The typing rules in LF adapt naturally to this difference with perhaps
the only noteworthy observation being that the rules for the binding
operators $\lambda$ and $\Pi$ introduce new nominal constants into the
typing context and correspondingly replace the bound variable by the
chosen nominal constant in the body.
We assume such an adaptation in the interpretation of these
judgements in the definition.
We also note here that the atomic formula $\fatm{G}{M}{A}$ includes
assertions of well-formedness for $G$ and $A$ in addition to the
derivability of $\lfchecktype{G}{M}{A}$.
\begin{definition}\label{def:validity}
  Let $F$ be a closed well-formed \logic{} formula.
  \begin{itemize}
    \item If $F$ is $\fatm{G}{M}{A}$, then it is valid exactly when
          $\lfwfctx{G}$, $\lfwftype{G}{A}$, and
          $\lfchecktype{G}{M}{A}$ are derivable.

    \item If $F$ is $\ftrue$ it is valid and if it is $\ffalse$ it is not valid.

    \item If $F$ is $\fimp{F_1}{F_2}$, it is valid if $F_1$ is not
          valid or $F_2$ is valid.

    \item If $F$ is $\fconjunct{F_1}{F_2}$, it is valid if both $F_1$ and
          $F_2$ are valid.

    \item If $F$ is $\fdisjunct{F_1}{F_2}$, it is valid if either $F_1$ or $F_2$ is valid.

    \item If $F$ is $\fctx{\Gamma}{\mathcal{C}}{F}$, it is valid if
          $\subst{F}{G/\Gamma}$ is valid for every $G$ such that
          $\csinst{}{\mathcal{C}}{G}$ is derivable.

    \item If $F$ is $\fall{x:\alpha}{F}$, it is valid if
          $\asub{F}{\asubentry{x}{M}{\alpha}}$ is valid for every $M$ such that
          $\stlctyjudg{\noms \cup \initctx}{M}{\alpha}$ is derivable.

    \item If $F$ is $\fexists{\typeof{x}{\alpha}}{F}$, it is valid if
          $\asub{F}{\asubentry{x}{M}{\alpha}}$ is valid for some $M$ such that
          $\stlctyjudg{\noms \cup \initctx}{M}{\alpha}$ is derivable.
  \end{itemize}
\end{definition}

\begin{example}
  \label{ex:plus-exist}
  To illustrate the capabilities of \logic{}, let us consider
  formalizing the property that the sum of two natural numbers is
  always defined.
  Based on the encoding presented in Figure~\ref{fig:size-sig}, this
  can be done through the formula
  \begin{tabbing}
    \qquad\=\qquad\=\qquad\=\kill
    \>$\fctx{\Gamma}{\mathcal{C}}{}\fall{\typeof{N_1}{\oty}}{}\fall{\typeof{N_2}{\oty}}{}$\\
    \>\>$\fimp{\fatm{\Gamma}{N_1}{\nat}}{}\fimp{\fatm{\Gamma}{N_2}{\nat}}{}$\\
    \>\>\>$\fexists{\typeof{N_3}{\oty}}{\fexists{\typeof{D}{\oty}}
        {\fatm{\Gamma}{D}{\plus{N_1}{N_2}{N_3}}}}$
  \end{tabbing}
  where $\mathcal{C}$, the context schema qualifying the
  quantification over $\Gamma$, is a context schema with a single empty block schema
  $\{\}()$.\footnote{A more natural statement is perhaps one that does
    not involve quantification over a context, a matter we discuss in
    Section~\ref{sec:applications}.}
  We can argue for the validity of this formula by induction on the height of
  the derivation of the judgement $\fatm{\Gamma}{N_1}{\nat}$ that
  appears in it.
  In more detail, the argument proceeds by case analysis on its
  derivation.
  In the case that it is derivable because $N_1$ is $\z$, we may pick
  the instantiations $N_2$ and $\plusz$ for $N_3$ and $D$,
  respectively, to show that the conclusion follows.
  When $N_1$ is $(\s\ N_1')$, we invoke the inductive hypothesis with
  respect to the derivability of $\fatm{\Gamma}{N'_1}{\nat}$ to conclude that
  $\lfchecktype{\Gamma}{D'}{\plus{N_1'}{N_2}{N_3'}}$ holds for some $N_3'$ and
  $D'$.
  We can then instantiate $N_3$ and $D$ to $(\s\ N_3')$ and $\pluss\ N_1'\ N_2\
    N_3'\ D'$, respectively, to complete the argument.
\end{example}

The context schema used in this example is the most natural one for
the property: the only instantiation for the quantified context
variable is the empty context $\cdot$, reflecting the fact that
natural numbers do not embody binding notions.
However, we can use richer context schemas that permit for non-empty
instantiations with bindings that do not impact the typing
judgements involved while still preserving the validity of the formula.
For example, the formula would still be valid if the context
variable quantification is governed by a
context schema that comprises a single block schema of the form
$x\of \tm, y : \size{x}{(\s\app z)}$.
In the next two sections, we will develop a sound, mechanizable method
for determining that formula validity is preserved under such changes
in context schema qualification.

\subsection{Proving the Validity of Formulas}\label{sec:proof-system}

The logic \logic{} is complemented by a proof system that provides the
basis for mechanizing arguments of validity for closed formulas; using
this system, it is possible to formalize arguments such as the one
described in Example~\ref{ex:plus-exist}.
The rules in the system enable the derivation of sequents of the
form $\seq[\mathbb{N}]{\Psi}{\Xi}{\Omega}{F}$, where $\mathbb{N}$
is a finite set of nominal constants, $\Psi$ is a finite set of term
variables with associated arity types, $\Xi$ is a finite set of
context variables with associated context variable types,
$\Omega$ is a finite set of \emph{assumption formulas}, and $F$ is a
\emph{conclusion} or \emph{goal} formula.
The formulas in $\Omega \cup \{F\}$ must be well-formed relative to
$\Psi$ and $\Xi$ and must use only those nominal constants that are
contained in $\mathbb{N}$; $\mathbb{N}$, $\Psi$, and $\Xi$ are
referred to as the \emph{support set}, the \emph{eigenvariables
context}, and the \emph{context variables context} of the sequent.
The goal of showing that a closed formula $F$ whose nominal
constants are contained in the set $\mathbb{N}$ is valid translates
into constructing a derivation for the sequent
$\seq[\mathbb{N}]{\emptyset}{\emptyset}{\emptyset}{F}$.
Using the quantifier rules will produce sequents with non-empty
eigenvariables and context variables context as proof obligations.
The context variable types will initially be equivalent to context
schemas but will take on a more elaborate structure in the treatment
of the typing of abstractions and case analysis over formulas of the
form $\fatm{G}{M}{A}$; we elide details since they are not relevant to
this paper.

Showing the soundness of the rules in the system requires a definition
of validity for sequents.
This is done first for closed sequents, \ie, ones of the form
$\seq[\mathbb{N}]{\emptyset}{\emptyset}{\Omega}{F}$:
this is valid if $F$ is valid or one of the
assumption formulas in $\Omega$ is not valid.
A sequent of the general form
$\seq[\mathbb{N}]{\Psi}{\Xi}{\Omega}{F}$ is then considered valid if
all of its closed instances are valid, where such instances are
obtained by substituting closed terms not containing the nominal
constants in $\mathbb{N}$ and respecting arity typing constraints for
the variables in $\Psi$ and replacing the variables in $\Xi$ with
closed context expressions that respect their context variable types.

\section{Context Schema Subsumption}
\label{sec:subsumption}

We can now state the objective of this paper precisely: we would like
to identify conditions under which the validity of a closed formula of
the form $\fctx{\Gamma}{\mathcal{C}}{F}$ implies that
of the formula $\fctx{\Gamma}{\mathcal{C}'}{F}$; the latter
formula represents the transportation of the former to the context
schema $\mathcal{C}'$.
For this formula to be valid, it must be the case that
$\ctxsub{F}{G'/\Gamma}$ is valid for every closed context expression
$G'$ that instantiates the context schema $\mathcal{C}'$.
In describing sufficient conditions for this property, we will
distinguish between instantiations for $\Gamma$ that constitute
well-formed context expressions with respect to the ambient signature
and those that do not.
For instantiations of the latter kind, we will impose a requirement on
the structure of the formula $F$.
We discuss this requirement in detail in the next section.
We focus in this section on the case of well-formed contexts.
For instantiations of this kind, we will identify a syntactically
checkable relation between $\mathcal{C}$ and $\mathcal{C'}$ that, when
combined with the validity of $\fctx{\Gamma}{\mathcal{C}}{F}$, will
ensure the validity of $\ctxsub{F}{G'/\Gamma}$ for every well-formed
instance $G'$ of $\mathcal{C}'$.

The particular relation that we will identify between the context
schemas $\mathcal{C}$ and $\mathcal{C}'$ is intended to have the
following content: if this relation holds, then, corresponding to each
well-formed instance $G'$ of $\mathcal{C}'$ there must exist a
well-formed instance $G$ of $\mathcal{C}$ that is such that $F$ is
valid under the substitution of $G'$ for $\Gamma$ exactly when $F$ is
valid under the substitution of $G$ for $\Gamma$.
At an intuitive level, this is a subsumption relation between the
context schemas $\mathcal{C}$ and $\mathcal{C}'$ that is parameterized
by the context variable $\Gamma$ and the formula $F$.
We denote the relation by
$\schemarelformula{\mathcal{C}}{F}{\mathcal{C}'}$.

We develop this relation and establish the essential property about it
in the rest of this section.
The first step in this direction is to identify a subsumption relation
between context expressions.
We do that in the first subsection.
In relating context schemas, we will need to match up block schemas
that comprise them.
This process is most easily expressed when the names of the variables
in block schemas are aligned.
To facilitate this, we digress briefly to describe the idea of a
variant of a block schema that enables the renaming of variables.
The last subsection utilizes the preceding machinery to define the
subsumption relation between context schemas and to show that it
embodies the desired property.

\subsection{Context Expression Subsumption}\label{ssec:cesubsumption}

Given two closed context expressions $G$ and $G'$, we desire to
characterize the situations in which $G$ contains all the information
that is needed from $G'$ to determine the validity of a formula $F$
when these context expressions are viewed as substitutions for a
context variable  $\Gamma$ that possibly appears free in $F$.
We denote this relation by $\cesub{\Gamma}{G}{G'}{F}$.
The first step in defining the relation is to identify those
situations in which an object of a type $A$ that appears in the
instantiation of $\Gamma$ can impact a typing judgement that appears
within $F$.
In this and other ensuing discussions, we assume the obvious extension
of the subordination relation to types in which nominal constants
appear.

\begin{figure}[tbhp]
  \centering

  \begin{tabular}{cc}

    $ \Infer
        { \tfsubord{A}{\fctx{\Gamma'}{\mathcal{C}}{F}}{\Gamma} }
        { \tfsubord{A}{F}{\Gamma} } \quad \Gamma \neq \Gamma' $

    &
    \quad

    $ \Infer
        { \tfsubord{A}{\fgeneric{x \!:\! \alpha}{F}}{\Gamma} }
        { \tfsubord{A}{F}{\Gamma} } \quad \genericq \in \{\forall, \exists\} $

  \end{tabular}

  \rowspace{}\rowspace{}

  \begin{tabular}{cc}

    $ \Infer
        { \tfsubord{A}{F_1 \bullet F_2}{\Gamma} }
        { \tfsubord{A}{F_1}{\Gamma} } \; \bullet \in \{\supset,\land,\lor\} $

    &

    $ \Infer
        { \tfsubord{A}{F_1 \bullet F_2}{\Gamma} }
        { \tfsubord{A}{F_2}{\Gamma} } \; \bullet \in \{\supset,\land,\lor\} $

  \end{tabular}

  \rowspace{}\rowspace{}

  \begin{tabular}{cc}

    $ \Infer
        { \tfsubord{A}{\fatm{\Gamma}{M}{A'}}{\Gamma} }
        { \subordinates{A}{A'} } $

    &

    $ \Infer
        { \tfsubord{A}
            {\fatm{\Gamma, n_1 \!:\! A_1, \ldots, n_k \!:\! A_k}{M}{A'}}
            {\Gamma} }
        { } $

  \end{tabular}

  \caption{Subordination of a Type by a Formula}\label{fig:tfsubord}
  \Description{Inferences rules determining when a type $A$ is subordinate to a
  formula relative to a context variable $\Gamma$. The rules specify that a type
  subordinates a formula with a quantifier ($\Pi, \forall, \exists$) precisely
  when the type subordinates the quantifier's body. A type subordinates a
  formula with a binary logical connective ($\land, \lor, \supset$) when the
  type subordinates at least one subformula. Finally, a type subordinates an
  atomic formula if context contains the context variable $\Gamma$ and it
  either: has a non-empty explicit context portion or the type $A$ is
  subordinate to the atomic formula's type $A'$.}
\end{figure}

\begin{definition}
The relation $\tfsubord{A}{F}{\Gamma}$, to be read as ``the type $A$ is
subordinate to the formula $F$ relative to the context variable
$\Gamma$,'' holds exactly when it is derivable by virtue of the rules in
Figure~\ref{fig:tfsubord}.
Note that the last rule applies only if there is at least one explicit
binding in the context expression that begins with $\Gamma$.
It is easy to see that $\tfsubord{A}{F}{\Gamma}$ is a decidable relation.
We write $\nottfsubord{A}{F}{\Gamma}$ to mean that
$\tfsubord{A}{F}{\Gamma}$ is not derivable.
\end{definition}

\begin{figure}[tbhp]
  \centering
  \begin{tabular}{ccc}

    $ \Infer
        { \cesub{\Gamma}{\emptyce}{\emptyce}{F} }
        { } $

    \quad
    &

    $ \Infer
        { \cesub{\Gamma}{G,\,n\!:\!A}{G',\,n\!:\!A}{F} }
        { \cesub{\Gamma}{G}{G'}{F} } $

    \quad
    &

    $ \Infer
        { \cesub{\Gamma}{G}{G',\,n\!:\!A}{F} }
        { \cesub{\Gamma}{G}{G'}{F} \qquad \nottfsubord{A}{F}{\Gamma} } $

  \end{tabular}

  \caption{Context Expression Subsumption Relative to a Formula and
    Context Variable}
  \label{fig:cesub}
  \Description{Described in text.}
\end{figure}

The desired relation between context expressions then amounts
to noting that it is possible to leave out bindings where the type is
not subordinate to the formula relative to the context variable for
which they are being contemplated as substitutions.

\begin{definition}
The relation $\cesub{\Gamma}{G}{G'}{F}$ holds between two context
expressions $G$ and $G'$, a context variable $\Gamma$, and a formula
$F$ exactly when this relation is derivable by the rules shown in
Figure~\ref{fig:cesub}.
Anticipating a necessity in Section~\ref{ssec:cssubsumption}, we
extend this relation to the situation where $G$ and $G'$ are block
declarations by letting bindings be of the form $x: A$ instead of
$n:A$.
\end{definition}

The following lemma captures the essence of the subsumption relation
for context expressions.
\begin{lemma}\label{lem:tfsubmin}
  If a formula of the form $\fatm{\Gamma}{M}{A'}$ occurs in the
  formula $F$ and $\cesub{\Gamma}{G}{G'}{F}$
  is derivable, then $G\subordmin{A'}$ and $G'\subordmin{A'}$ must be
  identical.
  If a formula of the form $\fatm{\Gamma,n_1 : A_1, \ldots, n_k :
    A_k}{M}{A'}$ occurs in $F$ and $\cesub{\Gamma}{G}{G'}{F}$
  is derivable, then $G$ must be identical to $G'$.
\end{lemma}

\begin{proof}
  In the first case, $\tfsubord{A}{F}{\Gamma}$ is
  derivable for any type $A$ such that $\subordinates{A}{A'}$.
  In the second case, $\tfsubord{A}{F}{\Gamma}$ is derivable for any
  type $A$.
  Using these observations, we can show both properties by an
  induction on the   definition of $G\subordmin{A'} =
  G'\subordmin{A'}$.
\end{proof}

The key property pertaining to the subsumption relation is expressed
in the following theorem.

\begin{theorem}\label{thm:wf-ctx-atm-transport}
  Let $F$ be a formula that is well-formed with respect to the type
  assignments $\Psi$ and $\Xi$ to term and context variables
  respectively.
  Further, let $\theta$ be a closed substitution for all the
  variables assigned types by $\Psi$ that is such that
  $\asub{F}{\theta}$ is defined, let $\Gamma$ be a context variable
  that possibly appears free in $F$, and let $\sigma$ be a closed
  substitution for all context variables other than $\Gamma$ that are
  assigned types by $\Xi$.
  Then for any closed context expressions $G$ and $G'$ such that
  $\lfwfctx{G}$, $\lfwfctx{G'}$ and $\cesub{\Gamma}{G}{G'}{F}$ are
  derivable, it is the case that
  $\ctxsub{\actxsub{F}{\theta}{\sigma}}{G/\Gamma}$ is valid if and
  only if
  $\ctxsub{\actxsub{F}{\theta}{\sigma}}{G'/\Gamma}$ is valid.
\end{theorem}

\begin{proof}
  The proof proceeds by induction on the structure of $F$.
  If $F$ is $\fdisjunct{F_1}{F_2}$, $\fconjunct{F_1}{F_2}$, $\fimp{F_1}{F_2}$,
  the substitutions distribute over the connective and the induction
  hypothesis easily yields the desired result.
  If $F$ is $\fall{x\colon\alpha}{F'}$ or
  $\fexists{x\colon\alpha}{F'}$, we may assume without loss of
  generality that $x$ is distinct from the free variables of $F$.
  But now we observe that the desired result would hold if
  $\ctxsub{\actxsub{F}{\theta\cup \{\langle x, M,
  \alpha\rangle\}}{\sigma}}{G/\Gamma}$ is valid if and only if
  $\ctxsub{\actxsub{F}{\theta\cup \{\langle x, M,
  \alpha\rangle\}}{\sigma}}{G'/\Gamma}$ is valid for any term $M$ of
  the appropriate arity type.
  However, this must be true by virtue of the induction hypothesis.
  An argument that is similar in structure applies when $F$ is of the
  form $\fctx{\Gamma'}{\mathcal{C}}{F'}$ and $\Gamma \neq \Gamma'$.
  The desired result follows trivially when $F$ is $\top$, $\bot$, of
  the form $\fctx{\Gamma}{\mathcal{C}}{F'}$ or
  an atomic formula of the form $\fatm{\Gamma', n_1:A_1, \ldots,
    n_k:A_k}{M}{A}$ or
  $\fatm{\Gamma'}{M}{A}$ where $\Gamma \ne \Gamma'$ because the
  formula is not affected by the substitution for $\Gamma$.

  It is only left to consider the cases where $F$ is an atomic
  formula of the form $\fatm{\Gamma, n_1 : A_1,\ldots, n_k :
    A_k}{M}{A}$ or $\fatm{\Gamma}{M}{A}$.
  In the former case, $G$ must be identical to $G'$ by
  Lemma~\ref{lem:tfsubmin} and the desired conclusion follows
  immediately.
  In the latter case, $\actxsub{F}{\theta}{\sigma}$ is of the form
  $\fatm{\Gamma}{\asub{M}{\theta}}{\asub{A}{\theta}}$.
  Since $\cesub{\Gamma}{G}{G'}{F}$ is derivable, it follows using
  Lemma~\ref{lem:tfsubmin} that $G\subordmin{A} = G'\subordmin{A}$.
  The desired result now follows from
  Proposition~\ref{thm:trans-canon-form}, the fact that the head of a type
  remains unchanged under substitution, and the assumptions that $G$
  and $G'$ are well-formed.
\end{proof}

\subsection{Variants of a Block Schema}

We now describe a process for renaming the variables that parameterize
a block schema as well as the ones to which it assigns types.
The key requirements of such a renaming is that it must not affect
well-formedness and it must leave the instances of the block schema
unchanged.
At a conceptual level, the renaming process is straightforward: we
simply apply a permutation of variable names to the block
schema.
However, the technical details are complicated by the
fact that the application of the permutation must be capture avoiding
and subsequent substitutions must be factored through the
permutation.
We handle these issues below by formulating the permutation as a
substitution.
The trusting reader may skip these details after noting the important
properties for block schema variants embodied in the statements of
Theorems~\ref{thm:variant-well-formed} and \ref{thm:variant-instance}.

In the discussions that follow, we say that a block declaration
$y_1:A_1,\ldots,y_m:A_m$ or a context expression
$n_1:A_1,\ldots,n_m:A_m$ is well-formed with respect to an arity
context $\STLCGamma$ if it is the case that $\wftype{\STLCGamma}{A_i}$
is derivable for $1 \leq i\leq m$.
\begin{definition}
  A variable permutation is a bijection on variables that differs from
  the identity map at only a finitely many points.
  We write $\{z_1/x_1,\ldots,z_n/x_n\}$ to denote the permutation that
  maps $x_i$ to $z_i$ for $1\leq i \leq n$ and that is the identity
  for all other variables.
  Let $\pi$ be a variable permutation.
  Then
  \begin{enumerate}
    \item $\pi.x$ denotes the result of applying $\pi$ to the variable
          $x$;

    \item the permutation substitution induced by an arity context
          $\STLCGamma$ from $\pi$, denoted by
          $\permsubs{\pi}{\STLCGamma}$, is
          \begin{tabbing}
            \quad\=$\{\langle x,y,\alpha'\rangle\ |\ $\=\kill
            \> $\{\langle x,z,\alpha'\rangle\ |\ z/x \in
              \pi\ \mbox{and}\ \alpha'\ \mbox{is}\ \alpha$\\
            \>\>$\mbox{if}\ x:\alpha
              \in \STLCGamma\ \mbox{and}\ \oty\ \mbox{otherwise}\}$;
          \end{tabbing}
    \item if $\Delta = y_1:A_1,\ldots,y_m:A_m$ is a block declaration
          that is well-formed with respect to the arity context
          $\STLCGamma$, then
          $\apppsubsbd{\pi}{\STLCGamma}{\Delta}$ denotes the block
          declaration
          \begin{tabbing}
            \quad\=\kill
            \>$\pi.y_1:\asub{A_1}{\permsubs{\pi}{\STLCGamma}},\ldots,
              \pi.y_m:\asub{A_m}{\permsubs{\pi}{\STLCGamma}}$;
          \end{tabbing}

        \item if $\mathcal{B}$ is the block schema
          \begin{tabbing}
            \quad=\kill
          $\{x_1:\alpha_1,\ldots,
            x_n:\alpha_n\}y_1:A_1,\ldots,y_m:A_m$
          \end{tabbing}
          then
          $\blkctx{\mathcal{B}}$ denotes the arity context
          \begin{tabbing}
            \quad\=\kill
            \>$\initctx \cup \{x_1:\alpha_1,\ldots,x_n:\alpha_n, y_1:\erase{A_1},\ldots,y_m:\erase{A_m}\};$
          \end{tabbing}

    \item if $\mathcal{B}$ is the block schema
          $\{x_1:\alpha_1,\ldots, x_n:\alpha_n\}\Delta$, then
          $\pi.\mathcal{B}$
      denotes the block schema
      \begin{tabbing}
        \quad\=\kill
        $\{\pi.x_1:\alpha_1,\ldots, \pi.x_n:\alpha_n\}
        \apppsubsbd{\pi}{\blkctx{\mathcal{B}}}{\Delta}.$
      \end{tabbing}
  \end{enumerate}
  The definition of $\apppsubsbd{\pi}{\STLCGamma}{\Delta}$ and
  $\pi.\mathcal{B}$ assumes that relevant hereditary
  substitutions are well-defined, a requirement that can be seen to hold
  under the respective well-formedness assumptions.
  For any permutation $\pi$, we say that $\pi.\mathcal{B}$
  is a variant of $\mathcal{B}$.
\end{definition}

The important properties of block schema variants are contained in the
following two theorems.

\begin{theorem}\label{thm:variant-well-formed}
  A variant of a well-formed block schema $\mathcal{B}$ is
  well-formed: if $\abstyping{\mathcal{B}}$ has a derivation then
  $\abstyping{\pi.\mathcal{B}}$ must also have one for any variable
  permutation $\pi$.
\end{theorem}

\begin{proof}
  If $\STLCGamma$ is an arity context, let
  $\pi.\STLCGamma$ represent the arity context obtained
  by replacing each assignment of the form $x:\alpha$ in $\STLCGamma$
  where $x$ is a variable by $\pi.x:\alpha$.
  We first show that if $\Delta$ is a block declaration that is
  well-formed with respect to the arity context $\STLCGamma''$ and it
  is the case that
  $\wfdecls{\STLCGamma}{\Delta}{\STLCGamma'}$ has a derivation then so
  must $\wfdecls{\pi.\STLCGamma}
    {\apppsubsbd{\pi}{\STLCGamma''}{\Delta}}
    {\pi.\STLCGamma'}$.
  We prove this by an induction on the derivation of
  $\wfdecls{\STLCGamma}{\Delta}{\STLCGamma'}$, making use of the
  fact that the arity kinding relation is preserved under
  permutations of variables.
  The theorem follows easily from this observation.
\end{proof}

\begin{theorem}\label{thm:variant-instance}
  Instances of block schemas are preserved across variants: if
  $\bsinst{}{\mathcal{B}}{G}$ has a derivation then so must
  $\bsinst{}{\pi.\mathcal{B}}{G}$ for any variable permutation $\pi$.
\end{theorem}

\begin{proof}
  Given a variable permutation $\pi$ and a substitution $\theta$,
  let $\pi.\theta$ denote the substitution
  $\{\langle \pi.x, t,\alpha\rangle\ |\ \langle x,t,\alpha\rangle \in
    \theta\}.$
  Let $\Delta$ be a block declaration that is well-formed with
  respect to $\STLCGamma$, let $G$ be a context expression, and let
  $\theta$ be a substitution such that $\declinst{}{\Delta}{G}{\theta}$ has a
  derivation.
  We then show by induction on this derivation that
  $\declinst{}{\apppsubsbd{\pi}{\STLCGamma}{\Delta}}{\apppsubsbd{\pi}{\STLCGamma}{G}}{\pi.\theta}$
  must have a derivation; the notation
  $\apppsubsbd{\pi}{\STLCGamma}{G}$ represents an operation that
  replaces each assignment $n:A$ in $G$ by $n: \asub{A}{\permsubs{\pi}{\STLCGamma}}$.
  Next, for a substitution $\rho$, let $\ctx{\rho}$ be the arity
  typing context $\{ x : \alpha \mid \asubentry{x}{t}{\alpha} \in \rho \}$.
  We then observe that if $G$ is well-formed with respect to
  $\initctx \cup \ctx{\rho}$, then
  $\apppsubsbd{\pi}{\STLCGamma}{G}$ is well-formed with respect to
  $\initctx \cup \ctx{\pi.\rho}$.
  We can now ascertain that
  $\asub{\apppsubsbd{\pi}{\STLCGamma}{G}}{\pi.\rho}$ is defined and
  must be identical to $\asub{G}{\rho}$.
  The desired result follows easily from these observations.
\end{proof}

\subsection{Subsumption for Context Schemas}\label{ssec:cssubsumption}

We now turn to defining the relation
$\schemarelformula{\mathcal{C}}{F}{\mathcal{C}'}$.
The definition we provide for it will have a constructive nature:
given a well-formed instance $G'$ of $\mathcal{C}'$, it will guide us
in constructing a well-formed instance $G$ of $\mathcal{C}$ that is such
that $F$ is valid under the substitution of $G'$ for $\Gamma$ exactly
when $F$ is valid under the substitution of $G$ for $\Gamma$.
More specifically, our definition of this relation will provide us a
means for ``pruning'' the context expression $G'$ to get $G$.
Given the way context schema instances are generated, this ability
will ultimately be based on a correspondence at the level
of the block schemas that comprise context schemas.
One part of the correspondence must pay attention to how context
expressions impact the derivability of typing judgements within $F$.
For this, the notion of context expression subsumption is useful.
However, we also need to pay attention to the impact of the pruning on
the well-formedness of the resulting context expression.
In particular, the pruning in an earlier part of a context expression
should not remove a binding that is necessary for the well-formedness
of a later part of the same expression.
The relation between block
declarations relative to a context schema that is defined below
provides the basis for ensuring this property.

\begin{definition}\label{def:ctxwftransport}
 Let $\mathcal{C}$ be a context schema and let $\Delta$ and $\Delta'$
 be block declarations.
 Then the relation $\ctxwftransport{\Delta}{\Delta'}$ holds exactly
 when it is derivable using the rules in
 Figure~\ref{fig:ctxwftransport}.
 We extend this relation to the situation where $\Delta$ and $\Delta'$
 are context expressions by letting bindings be of the form $n : A$
 rather than $x: A$ and allowing nominal constants to appear in
 types.
\end{definition}

\begin{figure}[tbhp]
  \centering
  \begin{tabular}{cc}
    $ \Infer
        { \ctxwftransport{\emptyce}{\emptyce} }
        { } $

    &
    \quad

    $ \Infer
        { \ctxwftransport{\Delta,\,x\!:\!A}{\Delta',\,x\!:\!A} }
        { \ctxwftransport{\Delta}{\Delta'} } $

  \end{tabular}

  \rowspace{}\rowspace{}

  \begin{tabular}{c}

    $ \Infer
        { \ctxwftransport{\Delta}{\Delta',\,x\!:\!A} }
        { \ctxwftransport{\Delta}{\Delta'} \qquad
          A \not\subord A' \ \text{for every } A' \
          \text{such that } {y\!:\!A'} \ \text{appears in } \mathcal{C} } $

  \end{tabular}

  \caption{A pruning relation for block declarations relative to a
    context schema}
  \label{fig:ctxwftransport}
  \Description{Inference rules which specify the pruning of a block
  declaration relative to a context schema.}
\end{figure}

The following lemma shows that the pruning operation
on context expressions just described will not impact the
well-formedness of subsequent assignments generated from the context schema
$\mathcal{C}$.

\begin{lemma}\label{lem:ctxmin}
  Let $\mathcal{C}$ be a context schema and let $G$ and $G'$ be
  context expressions such that $\ctxwftransport{G}{G'}$.
  Further, let $x:A$ be an assignment in a block declaration in
  $\mathcal{C}$.
  Then $G\subordmin{A} = G'\subordmin{A}$.
\end{lemma}

\begin{proof}
  By an induction on the structure of $G'$. We omit the details.
\end{proof}

\begin{figure*}[tbhp]
  \centering

  \begin{tabular}{cc}

    $ \Infer
        { \blockschemarelformula{\mathcal{B}'}{F}{\mathcal{C},\mathcal{B}} }
        { \mathcal{B}' = \{x_{1}:\alpha_{1}, \ldots, x_{n}: \alpha_{n}\}\,\Delta' \qquad
          \mathcal{B} = \{y_{1}:\alpha_{1}, \ldots, y_{k}: \alpha_{k}\}\,\Delta \qquad
          \blockwftransport[\mathcal{C}]{\Delta}{\Delta'} \qquad
          \cesub{\Gamma}{\Delta}{\Delta'}{F} } $

    \quad & \quad

    $ \Infer
        { \blockschemarelformula{\mathcal{B}'}{F}{\mathcal{C},\mathcal{B}} }
        { \blockschemarelformula{\mathcal{B}'}{F}{\mathcal{C}} } $

  \end{tabular}

%  \vspace{5pt}
  \rowspace{}\rowspace{}

  \begin{tabular}{cc}

    $ \Infer
        { \schemarelformula{\mathcal{C}}{F}{\mathcal{C}',\mathcal{B}} }
        { \schemarelformula{\mathcal{C}}{F}{\mathcal{C}'} \qquad
          \mathcal{B'}\ \text{is a variant of}\ \mathcal{B} \qquad
          \blockschemarelformula{\mathcal{B}'}{F}{\mathcal{C}} } $

    \qquad & \qquad

    $ \Infer
        { \schemarelformula{\mathcal{C}}{F}{\cdot} }
        { } $

  \end{tabular}

  \caption{A subsumption relation on context schemas}
  \label{fig:ctxsubsumption}
  \Description{Describe in the text.}
\end{figure*}

The next two lemmas together show that if $\mathcal{B}$ is a block
schema whose block declaration is a pruned version of that of
$\mathcal{B}'$ that also subsumes it, then there is an instance $G$ of
$\mathcal{B}$ that is a pruned version of $G'$ and subsumes $G'$ for
any context expression $G'$ that instantiates $\mathcal{B}'$.

\begin{lemma}\label{lem:relsbdinst}
  Let $\mathcal{C}$ be a context schema and let $\Delta$ and $\Delta'$
  be block declarations such that
  $\ctxsimfmla{\Delta}{\Delta'}{F}$ and
  $\ctxwftransport{\Delta}{\Delta'}$ are derivable.
  Further, let $\mathbb{N}$, $G'$ and $\theta'$ be such that
  $\declinst{\mathbb{N}}{\Delta'}{G'}{\theta'}$ is derivable.
  Then there must be some $G$ and $\theta$ such that
  $\declinst{\mathbb{N}}{\Delta}{G}{\theta}$,
  $\ctxsimfmla{G}{G'}{F}$, and
  $\ctxwftransport{G}{G'}$ are derivable.
\end{lemma}

\begin{proof}
  By induction on the derivation of
  $\declinst{\mathbb{N}}{\Delta'}{G'}{\theta'}$; the generated $G$
  will lose some assignments to nominal constants present in $G'$ and
  the $\theta$ will correspondingly lose some substitutions from
  $\theta'$.
  The needed relationships between $G$ and $G'$ will follow easily
  from the given ones between $\Delta$ and $\Delta'$ and the fact that
  the subordination relation between type expressions is invariant
  under substitutions.
\end{proof}

\begin{lemma}\label{lem:relssubst}
  Let $\mathcal{C}$ be a context schema and let $G$ and $G'$ be
  context expressions such that   $\ctxsimfmla{G}{G'}{F}$ and
  $\ctxwftransport{G}{G'}$ are derivable.
  For any substitution $\theta$ if $\asub{G'}{\theta}$ is
  defined then so must $\asub{G}{\theta}$ be.
  Moreover, $\ctxsimfmla{\asub{G}{\theta}}{\asub{G'}{\theta}}{F}$
  and $\ctxwftransport{\asub{G}{\theta}}{\asub{G'}{\theta}}$
  must be derivable.
\end{lemma}

\begin{proof}
  The first property, that $\asub{G}{\theta}$ is defined, follows from
  observing that every type that appears in $G$ also appears in $G'$.
  The remaining properties follow from the fact that the subordination
  between types is not impacted by substitutions.
\end{proof}

We are now in a position to define the subsumption relation between
context schemas.

\begin{definition}
  The relation
  $\schemarelformula[\Gamma]{\mathcal{C}}{F}{\mathcal{C}'}$, to be
  read as ``the context schema $\mathcal{C}$ subsumes the context schema
  $\mathcal{C}'$ relative to the formula $F$ and the context variable
  $\Gamma$,'' holds exactly when it is derivable by virtue of the rules in
  Figure~\ref{fig:ctxsubsumption}.
\end{definition}

A key part of this definition is the relation
$\blockschemarelformula[\Gamma]{\mathcal{B}'}{F}{\mathcal{C}}$ which
matches block schemas in  $\mathcal{C}'$ with ones in $\mathcal{C}$.
Intuitively, a block schema from $\mathcal{C}'$ matches one from
$\mathcal{C}$ if bindings from a context expression that results from
it can be pruned to get a context expression from the matching block
schema, and the pruning will not impact the well-formedness of other
context expressions generated from block schemas in $\mathcal{C}$ or
typing judgements in $F$ whose context is determined by $\Gamma$.
The next two lemmas and culminating theorem formalize this intuition.

\begin{lemma}\label{lem:oneinsttransport}
  Let $F$ be a formula in which the context variable $\Gamma$
  possibly appears free and let $\mathcal{C}$ and $\mathcal{C}'$ be
  context schemas such that
  $\schemarelformula{\mathcal{C}}{F}{\mathcal{C}'}$ is derivable.
  Further, let $G'$ be a context expression such that
  $\csinstone{}{\mathcal{C}'}{G'}$ is derivable.
  Then there is a context expression $G$ such that
  $\csinstone{}{\mathcal{C}}{G}$,
  $\ctxsimfmla{G}{G'}{F}$, and
  $\ctxwftransport{G}{G'}$ are derivable.
\end{lemma}

\begin{proof}
  Since $\csinstone{}{\mathcal{C}'}{G'}$, it must be
  the case that there is a derivation for
  $\bsinst{}{\mathcal{B}''}{G'}$ for some block schema
  $\mathcal{B}''$ in $\mathcal{C}'$.
  Since $\schemarelformula{\mathcal{C}}{F}{\mathcal{C}'}$ is
  derivable, there must be a variant $\mathcal{B}'$ of $\mathcal{B}''$
  of the form $\{y_1:\alpha_1,\ldots, y_n:\alpha_n\}\Delta'$ and
  a block schema $\mathcal{B}$ of the form
  $\{x_1:\alpha_1,\ldots, x_m:\alpha_m\}\Delta$ in $\mathcal{C}$ such
  that $\ctxsimfmla{\Delta}{\Delta'}{F}$ and
  $\ctxwftransport{\Delta}{\Delta'}$ are derivable.
  By Theorem~\ref{thm:variant-instance}, there must be a derivation
  for $\bsinst{}{\mathcal{B}'}{G'}$.
  But then it must be the case that there is a $G'_1$ and a $\rho'$ such
  that $\declinst{}{\Delta'}{G_1'}{\rho'}$ is derivable
  and a $\theta$ such that $G'= \asub{G'_1}{\theta}$.
  By Lemma~\ref{lem:relsbdinst}, there must be a context expression
  $G_1$ and a $\rho$ such that
  $\declinst{}{\Delta}{G_1}{\rho}$, $\ctxsimfmla{G_1}{G'_1}{F}$ and
  $\ctxwftransport{G_1}{G'_1}$ are derivable.
  Using these facts and Lemma~\ref{lem:relssubst}, it is easily
  verified that $\asub{G_1}{\theta}$ is defined and that
  $\bsinst{}{\mathcal{B}}{\asub{G_1}{\theta}}$ and,
  hence, $\csinstone{}{\mathcal{C}}{\asub{G_1}{\theta}}$ is derivable.
  Letting $G$ be $\asub{G_1}{\theta}$ and using
  Lemma~\ref{lem:relssubst} again, we can check that the requirements
  of the lemma hold.
\end{proof}

\begin{lemma}\label{lem:joinctxs}
  Let $F$ be a formula, let $\Gamma$ be a context variable
  that possibly appears in $F$, and let $\mathcal{C}$ and $\mathcal{C}'$ be
  context schemas such that
  $\schemarelformula{\mathcal{C}}{F}{\mathcal{C}'}$ is derivable.
  Let $G$ and $G'$ be context expressions such that $\lfwfctx{G}$,
  $\lfwfctx{G'}$, $\ctxwftransport{G}{G'}$ and $\ctxsimfmla{G}{G'}{F}$
  are derivable.
  Let $G'_1$ be a context expression such that $\lfwfctx{G',G'_1}$ is
  derivable.
  Let $G_1$ be a context expression such that for every assignment of
  the form $n:A$ in $G_1$ it is the case that $A$ is a substitution
  instance of some type $B$ such that an assignment of the form $x:B$
  appears in one of the block declarations in $\mathcal{C}$.
  Further, let $\ctxsimfmla{G_1}{G'_1}{F}$ and
  $\ctxwftransport{G_1}{G'_1}$ be derivable.
  Then $\lfwfctx{G,G_1}$;
  $\ctxsimfmla{G,G_1}{G',G'_1}{F}$; and
  $\ctxwftransport{G,G_1}{G',G'_1}$ must be derivable.
\end{lemma}

\begin{proof}
  By induction on the structure of $G'_1$.
  If $G'_1$ is $\emptybb$, the conclusion follows directly from the
  assumptions.
  Suppose, then, that $G'_1$ is of the form $G'_2, n : A$.
  Note that it must be the case that $\lfwfctx{G',G'_2}$.
  Now, $G_1$ could either include the assignment for $n$ or not. We
  consider each of these cases separately below.

  Let us first consider the case where $G_1$ includes the assignment
  for $n$, \ie, when it is of the form $G_2, n : A$.
  Here, it must be the case that $\ctxsimfmla{G_2}{G'_2}{F}$ and
  $\ctxwftransport{G_2}{G'_2}$ are derivable.
  By the induction hypothesis, $\ctxsimfmla{G,G_2}{G',G'_2}{F}$;
  $\ctxwftransport{G,G_2}{G', G'_2}$; and $\lfwfctx{G,G_2}$ must be
  derivable.
  The derivability of the first two judgements implies
  that the judgements
  $\ctxsimfmla{G,G_2,n:A}{G',G'_2,n:A}{F}$ and
  $\ctxwftransport{G,G_2,n:A}{G', G'_2,n:A}$ must be derivable.
  Since there is a derivation for $\lfwfctx{G',G'_2,n:A}$, it must be
  the case that $n$ does not appear in $G',G'_2$ and that
  $\lfwftype{G',G'_2}{A}$ has a derivation.
  Clearly, $n$ must not appear in $G,G_2$ either since this
  comprises a subcollection of the bindings in $G',G'_2$.
  From Lemma~\ref{lem:ctxmin}, we see that $(G,G_2)\subordmin{B}$ is
  identical to $(G',G'_2)\subordmin{B}$ for any type $B$ such that
  $x:B$ appears in a block declaration in $\mathcal{C}$.
  Using Proposition~\ref{thm:trans-canon-form} and the fact that the
  head of a type remains unchanged under substitution, it follows that
  $\lfwftype{G,G_2}{A}$ must have a derivation if there is one for
  $\lfwftype{G',G'_2}{A}$.
  That $\lfwfctx{G,G_2, n: A}$ has a derivation is an easy
  consequence.

  The argument in the case when $G_1$ does not include the assignment
  for $n$ is simpler and, in fact, follows immediately from the
  induction hypothesis.
  We omit the details.
\end{proof}

\begin{theorem}\label{thm:subsumes}
  Let $F$ be a formula in \logic, let $\Gamma$ be a context variable
  that possibly appears in $F$, and let $\mathcal{C}$ and $\mathcal{C}'$ be
  context schemas such that
  $\schemarelformula{\mathcal{C}}{F}{\mathcal{C}'}$ is derivable.
  Let $G'$ be a context expression such that
  $\csinst{}{\mathcal{C}'}{G'}$ and
  $\lfwfctx{G'}$ are derivable.
  Then there is a context expression $G$ such that
  $\csinst{}{\mathcal{C}}{G}$, $\lfwfctx{G}$,
  $\ctxsimfmla{G}{G'}{F}$, and $\ctxwftransport{G}{G'}$ are derivable.
\end{theorem}

\begin{proof}
  By an induction on the derivation of
  $\csinst{}{\mathcal{C}'}{G'}$.

  In the base case, $G'$ is $\emptyce$.
  The argument in this case is trivial: we simply observe that
  $\csinst{}{\mathcal{C}}{\emptyce}$ and then
  verify that the other three requirements are met by setting $G$ to
  $\emptyce$.

  Now suppose that $G'$ is of the form $G'_1, G'_2$ where
  $\csinst{}{\mathcal{C'}}{G'_1}$ and
  $\csinstone{}{\mathcal{C}'}{G'_2}$ are derivable;
  note that, from the assumption that $\lfwfctx{G'}$ it follows that
  $\lfwfctx{G'_1}$.
  By Lemma~\ref{lem:oneinsttransport} there is a context expression
  $G_2$ such that $\csinstone{}{\mathcal{C}}{G_2}$,
  $\ctxsimfmla{G_2}{G'_2}{F}$ and $\ctxwftransport{G_2}{G'_2}$ are
  derivable.
  By the induction hypothesis, there is a context expression $G_1$
  such that
  $\csinst{}{\mathcal{C}}{G_1}$, $\lfwfctx{G_1}$,
  $\ctxsimfmla{G_1}{G'_1}{F}$, and $\ctxwftransport{G_1}{G'_1}$ are
  derivable.
  From what we have at hand, it is clear that
  $\csinst{}{\mathcal{C}}{G_1,G_2}$ has a
  derivation.
  Let us now take $G$ to be $G_1, G_2$.
  All the pieces are then at hand for us to invoke
  Lemma~\ref{lem:joinctxs} and thereby complete the proof.
\end{proof}

Theorems~\ref{thm:subsumes} and \ref{thm:wf-ctx-atm-transport}
together ensure that our subsumption relation captures the property we
want of it: if $\schemarelformula{\mathcal{C}}{F}{\mathcal{C}'}$
holds, then corresponding to every well-formed instance $G'$ of
$\mathcal{C}'$ there is a well-formed instance $G$ of $\mathcal{C}$
such that $\ctxsub{F_1}{G'/\Gamma}$ is valid exactly when
$\ctxsub{F_1}{G/\Gamma}$ is valid, where $F_1$ is an instance of the
formula $F$ in which at most the variable $\Gamma$ appears free.

\begin{figure*}[tbhp]
  \centering

  \begin{tabular}{cccc}
    $ \Infer
        { \illformvalidant{\Gamma}{\bot} }
        { } $

    \qquad
    &
    \qquad

    $ \Infer
        { \illformvalidant{\Gamma}{\fatm{G}{M}{A}} }
        { \text{$\Gamma$ appears in } G } $

    \qquad
    &
    \qquad

    $ \Infer
        { \illformvalidant{\Gamma}{F_1 \supset F_2} }
        { \illformvalidcons{\Gamma}{F_1} \qquad
          \illformvalidant{\Gamma}{F_2} } $

    \qquad
    &
    \qquad

    $ \Infer
        { \illformvalidant{\Gamma}{F_1 \lor F_2} }
        { \illformvalidant{\Gamma}{F_1} \qquad
          \illformvalidant{\Gamma}{F_2} } $

  \end{tabular}

  \rowspace{}\rowspace{}

  \begin{tabular}{ccc}

    $ \Infer
        { \illformvalidant{\Gamma}{F_1 \land F_2} }
        { \illformvalidant{\Gamma}{F_i} } \;\,i \in \{1,2\}\, $

    \qquad
    &
    \qquad

    $ \Infer
        { \illformvalidant{\Gamma}{\mathcal{Q}x\!:\!\alpha.\,F'} }
        { \illformvalidant{\Gamma}{F'} } \;\,\mathcal{Q} \in \{\forall,\exists\}\, $

    \qquad
    &
    \qquad

    $ \Infer
        { \illformvalidant{\Gamma}{\fctx{\Gamma'}{\mathcal{C}}{F'}} }
        { \illformvalidant{\Gamma}{F'} } \;\,\Gamma' \neq \Gamma\, $

  \end{tabular}

  %  \rowsectionspace{}\rowspace{}\rowspace{}
    \rowspace{}\rowspace{}\rowspace{}

  \begin{tabular}{cccc}

    $ \Infer
        { \illformvalidcons{\Gamma}{\top} }
        { } $

    \qquad
    &
    \qquad

    $ \Infer
        { \illformvalidcons{\Gamma}{\fdisjunct{F_1}{F_2}} }
        { \illformvalidcons{\Gamma}{F_i} } \;\,i \in \{1,2\}\, $

    \qquad
    &
    \qquad

    $ \Infer
        { \illformvalidcons{\Gamma}{F_1 \land F_2} }
        { \illformvalidcons{\Gamma}{F_1} \qquad
          \illformvalidcons{\Gamma}{F_2} } $

    \qquad
    &
    \qquad

    $ \Infer
        { \illformvalidcons{\Gamma}{F_1 \supset F_2} }
        { \illformvalidant{\Gamma}{F_1} } $

  \end{tabular}

  \rowspace{}\rowspace{}

  \begin{tabular}{ccc}

    $ \Infer
        { \illformvalidcons{\Gamma}{F_1 \supset F_2} }
        { \illformvalidcons{\Gamma}{F_2} } $

    \qquad
    &
    \qquad

    $ \Infer
        { \illformvalidcons{\Gamma}{\mathcal{Q}x\!:\!\alpha.\,F'} }
        { \illformvalidcons{\Gamma}{F'} } \;\,\mathcal{Q} \in \{\forall,\exists\}\, $

    \qquad
    &
    \qquad

    $ \Infer
        { \illformvalidcons{\Gamma}{\fctx{\Gamma'}{\mathcal{C}}{F'}} }
        { \illformvalidcons{\Gamma}{F'} } \;\,\Gamma' \neq \Gamma\, $

  \end{tabular}

  \caption{Structural analysis for formula validity under an ill-formed substitution}
  \label{fig:illformed-validity}
  \Description{Described in text.}
\end{figure*}

\section{Transporting Theorems Across Context Schemas}\label{sec:rule-and-soundness}

The analysis in the previous section goes only part of the way towards
determining when we can transport a theorem of the form
$\fctx{\Gamma}{\mathcal{C}}{F}$ into one where the quantification of
the context variable is governed by a different schema $\mathcal{C}'$:
we still need to consider the validity of $F$ when $\Gamma$ is
substituted for by an ill-formed instance of $\mathcal{C}'$ for
$\Gamma$.
An approach based on subordination that is similar to the one in
Section~\ref{sec:subsumption} will not work in this situation:
removing parts of a context expression could well render the result
well-formed.
In light of this, we use a different idea: we identify a structural
property for formulas that ensures that they must be valid under
ill-formed substitutions for a context variable $\Gamma$.
We describe this property in the first subsection below and then
combine it with the results of the previous section to yield
a sound proof rule for transporting theorems across context schemas.

\subsection{Validity Under an Ill-formed Context Variable
  Substitution}
\label{sec:illformed-validity}

By the definition of validity for closed atomic formulas, a formula of
the form $\fatm{G}{M}{A}$ is not valid if $G$ is an ill-formed
context expression.
When such a formula appears negatively, it can provide the basis for
concluding that the overall formula must be valid.
We must, of course, pay attention to the overall propositional
structure of the formula to leverage this observation.
Figure~\ref{fig:illformed-validity} presents rules for inferring the
relation $\illformvalidcons{\Gamma}{F}$ and the relation
$\illformvalidant{\Gamma}{F}$ that embody this kind of structural
analysis of the formula $F$.
The following theorem quantifies the intent of such an analysis.

\begin{theorem}\label{thm:illformed-validity}
  Let $F$ be a formula that is well-formed with respect to the type
  assignments $\Psi$ and $\Xi$ to term and context variables
  respectively.
  Further, let $\theta$ be a closed substitution for all the
  variables assigned types by $\Psi$ that is such that
  $\asub{F}{\theta}$ is defined, let $\Gamma$ be a context variable
  that possibly appears free in $F$, and let $\sigma$ be a closed
  substitution for all context variables other than $\Gamma$ that are
  assigned types by $\Xi$.
  Finally, let $G$ be a closed context expression such that
  $\lfwfctx{G}$ is not derivable. Then

  \begin{enumerate}
    \item $\ctxsub{\actxsub{F}{\theta}{\sigma}}{G/\Gamma}$ is valid if
    $\illformvalidcons{\Gamma}{F}$ is derivable.
    \item $\ctxsub{\actxsub{F}{\theta}{\sigma}}{G/\Gamma}$ is not
      valid if $\illformvalidant{\Gamma}{F}$ is derivable.
  \end{enumerate}
\end{theorem}

\begin{proof}
  The two parts can be proved simultaneously by an
  induction on the structure of the formula $F$, using the definition
  of validity.
\end{proof}

\subsection{A Proof Rule for Transporting Theorems}

The desired proof rule for transporting theorems proved under one
context schema to another context schema can now be presented.
This is done in Figure~\ref{fig:proof-rule}.

\begin{figure}[htbp]
\begin{center}
\begin{tabular}{c}
$ \Infer
    { \seq{\Psi}{\Xi}{\Omega}{\fctx{\Gamma}{\mathcal{C}'}{F}} }
    { \seq{\Psi}{\Xi}{\Omega}{\fctx{\Gamma}{\mathcal{C}}{F}} \qquad
      \schemarelformula{\mathcal{C}}{F}{\mathcal{C}'} \qquad
      \illformvalidcons{\Gamma}{F} } $
\end{tabular}
\end{center}
\caption{The Transportation Proof Rule}\label{fig:proof-rule}
\end{figure}

The soundness of the proof rule is the content of the
following theorem. 
\begin{theorem}\label{thm:schema-transport}
  Suppose that
  $\seq{\Psi}{\Xi}{\Omega}{\fctx{\Gamma}{\mathcal{C}}{F}}$ is a
  well-formed sequent and that $\mathcal{C}'$ is well-formed context
  schema, such that $\schemarelformula{\mathcal{C}}{F}{\mathcal{C}'}$
  is derivable.
  Further, let $\illformvalidcons{\Gamma}{F}$ be derivable.
  If $\seq{\Psi}{\Xi}{\Omega}{\fctx{\Gamma}{\mathcal{C}}{F}}$ is
  a valid sequent, then the sequent
  $\seq{\Psi}{\Xi}{\Omega}{\fctx{\Gamma}{\mathcal{C}'}{F}}$ must also
  be valid.
\end{theorem}

\begin{proof}
  Follows from the definition of validity for sequents, using
  Theorems~\ref{thm:subsumes}, \ref{thm:wf-ctx-atm-transport}, and
  \ref{thm:illformed-validity}.
\end{proof}

\section{Using the Proof Rule in Reasoning}\label{sec:applications}

We first consider how the proof rule can be used to realize the
style of reasoning sketched in the introduction.
Our interest is in showing that every term has an associated size.
Assuming the LF declarations in Figure~\ref{fig:size-sig}, this
property is stated in \logic{} by the following formula:
\[
  \fall{\typeof{M}{o}}{}\fimp{\fatm{\cdot}{M}{\tm}}{}
  \fexists{\typeof{N}{o}}{}\fexists{\typeof{D}{\oty}}{}
  \fatm{\cdot}{D}{\size{M}{N}}.
\]
This formula can be proved by an induction over the
structure of the representation of an untyped lambda term.
Noting the way a typing judgement is derived for the
representation of an abstraction term, we see that we will actually
need to prove the following generalized version of the
formula that allows for non-empty typing contexts:
\[
  \fctx{\Gamma}{\mathcal{C}'}{}
  \fall{\typeof{M}{o}}{}\fimp{\fatm{\Gamma}{M}{\tm}}{}
  \fexists{\typeof{N}{o}}{}\fexists{\typeof{D}{\oty}}{}
  \fatm{\Gamma}{D}{\size{M}{N}},
\]
Here, $\mathcal{C}'$ is the context schema $\{\}(\typeof{x}{\tm},
  \typeof{y}{\size{x}{(\s\ \z)}})$.
The abstraction case utilizes weakening to introduce an entry for the size of
the associated term in the context of the antecedent, which preserves the height
of the derivation.
The case then follows easily from the induction hypothesis under this generalization.
In the application case, i.e., when $M$ has the form
$\tmapp\app M_1\app M_2$, the induction hypothesis allows us to
conclude that there are terms $D_1$, $D_2$, $N_1$ and $N_2$ such that
$\fatm{\Gamma}{D_1}{\size{M_1}{N_1}}$ and
$\fatm{\Gamma}{D_2}{\size{M_2}{N_2}}$ have derivations.
The derivability of these formulas also implies that there are
derivations for $\fatm{\Gamma}{N_1}{\nat}$ and
$\fatm{\Gamma}{N_2}{\nat}$.
Thus, if $\fctx{\Gamma}{\mathcal{C}'}{F}$ where $F$ is the formula
  \begin{tabbing}
    \qquad\=\qquad\qquad\qquad\=\kill
   \> $\fall{\typeof{N_1}{\oty}}
            {\fall{\typeof{N_2}{\oty}}
                  {\fimp{\fatm{\Gamma}{N_1}{\nat}}
                        {\fimp{\fatm{\Gamma}{N_2}{\nat}}}}}$\\
     \>\>$\fexists{\typeof{N_3}{\oty}}{\fexists{\typeof{D}{\oty}}
         {\fatm{\Gamma}{D}
         {\plus{N_1}{N_2}{N_3}}}}$
  \end{tabbing}
  were derivable, then, using the $\sizeapp$ constant, we would be
  able to provide terms $D_3$ and $N_3$ such that
  \begin{tabbing}
  \qquad\=\kill
  \>$\fatm{\Gamma}{D_3}{\size{(\tmapp\app M_1\app M_2)}{(\s\app
      N_3)}}$
  \end{tabbing}
  is derivable.

The proof obligation that we are left with looks much like the formula
in Example~\ref{ex:plus-exist}.
However, the context schema governing the quantification in the two cases is
different: in the formula in Example~\ref{ex:plus-exist}, which is what we would
expect to find in a ``library'' formalization of natural numbers, the
quantification is governed by the context schema $\{\}()$.
We could prove the property afresh with the new context schema but,
from an efficiency and modularity perspective, it is desirable to
obtain the conclusion directly from the library version.
We can use the proof rule for transporting theorems towards this end
provided we are able to show that $\illformvalidcons{\Gamma}{F}$ and
$\schemarelformula{\mathcal{C}}{F}{\mathcal{C}'}$ have derivations,
where $\mathcal{C}$ is the context schema $\{\}()$.
Both properties can be checked by a mechanizable process: for the
latter, the LF specifications must be analyzed to conclude
that
$\tm \not\subord \nat$, $\tm \not\subord \plussym$,  $\sizesym
\not\subord \nat$, and $\sizesym \not\subord \plussym$ hold.

We have modified the Adelfa proof assistant to incorporate support for
this kind of reasoning.
Rather than providing a separate tactic for the proof rule, we have
incorporated it into the \emph{apply} tactic.
This tactic utilizes premises $F_1,\ldots,F_n$ together with a
``lemma'' of the form
$\Pi\overline{\Gamma:\mathcal{C}}
   \fall{\overline{x:\alpha}}{\fimp{F'_1}
                             {\ldots\fimp{F'_n}{F'}}}$
to infer a suitable instance of the formula $F'$.
The matching of the premise formulas with instances of the antecedent
formulas in the lemma may require the adjustment of the context schema
and this is done automatically based on our proof rule in a new
version of the \emph{apply} tactic.

What we have considered above is the paradigmatic use
we envisage for our rule.
LF specifications often employ general relations (such as
\plussym) in the definition of specific relations (such as \sizesym).
Our rule permits theorems about the general relations proved without
reference to where they may be used to be drawn upon in more specific
reasoning tasks.
We sketch three other examples drawn from the literature to illustrate
such applications; complete proof developments can be found at the
website mentioned in the introduction.
The first concerns the task of showing that the size of a lambda term
does not decrease under substitution~\cite{twelf.website}.
The proof of this theorem must be relativized to typing contexts that
arise in assessing the size of a term.
An argument that is inductive on the structure of the term being
substituted into requires the following property about natural
numbers for the application case: if $n_1$ and $n_2$ are less than or
equal to $n'_1$ and $n'_2$ respectively, then $n_1 + n_2$ is less than
or equal to $n'_1 + n'_2$.
The typing context would be empty in an independent statement of this
property, but the gap can be filled using our
rule.
The second example is a ``benchmark'' property presented
in~\cite{felty15jar} that concerns showing declarative equality for
lambda terms implies algorithmic equality.
The difference between these two notions is that the former explicitly
builds in reflexivity, symmetry and transitivity.
These properties are of independent interest for algorithmic equality
and should be proved separately for it.
The proof rule for transporting theorems allows us to lift these
results to the context needed in establishing the desired equivalence.
The third example is a solution to problem 1a of the {\sc PoplMark}
challenge~\cite{aydemir05tphols} that emulates one provided using the
Twelf system~\cite{ashley-rollman05poplmark}.
This problem concerns proving transitivity and narrowing for the subtyping
relation in System $\texttt{F}_{<:}$ when this relation is specified
algorithmically.
For technical reasons that are orthogonal to the concerns in this
paper, it is most convenient to show these properties first using a
context schema with multiple block schemas.
However, it is desirable to present the final result relative to a
context schema with a single block schema.
The transportation mechanism that is supported by the proof rule
developed in this paper allows for this change to be realized
automatically.
This example is interesting because it involves the use of context
schema subsumption where one of the context schemas comprises multiple
block schemas.
Additionally, the example features the use of our ideas in a situation
where the context quantification is not at the outermost level in the
formula.

The most natural presentation for properties that do not involve a
typing context is to omit the context quantification.
Thus, the preferred presentation for the theorem about the existence
of a sum for two natural numbers is perhaps the following:
  \begin{tabbing}
    \qquad\=\qquad\qquad\qquad\=\kill
    \>$\fall{\typeof{N_1}{\oty}}
            {\fall{\typeof{N_2}{\oty}}
                  {\fimp{\fatm{\cdot}{N_1}{\nat}}
                        {\fimp{\fatm{\cdot}{N_2}{\nat}}}}}$\\
     \>\>${\fexists{\typeof{N_3}{\oty}}{\fexists{\typeof{D}{\oty}}{\fatm{\cdot}{D}{\plus{N_1}{N_2}{N_3}}}}}$
  \end{tabbing}
The use of such presentations can be supported without sacrificing the ability
to benefit from the transportation rule by essentially building the ``context
generalization'' into the \emph{apply} tactic.
We support this ability in our implementation.

\section{Conclusion}\label{sec:conclusion}

In this paper, we have described a method for transporting theorems
about LF specifications that have been proved relative to contexts
adhering to one schematic description to a situation where they adhere
to a different such description.
This method is based on the idea of context schema subsumption that
uses in an intrinsic way the previously described notion of
subordination for LF types~\cite{harper07jfp,virga99phd}.
We have developed a sound proof rule around context schema subsumption
and have illuminated its use through particular
reasoning examples.
We have also discussed the incorporation of the rule into the
Adelfa proof assistant.

In~\cite{felty15arxiv}, Felty \etal{} have described a collection of
benchmarks for systems that use higher-order abstract syntax in
reasoning tasks.
The problems addressed in this paper fall into the category they
describe as treating ``non-linear context extensions.''
The method we have developed in fact treats problems that need a
``generalized context'' approach under their definition.
The complementary category is that encompassed by a ``relational''
style of reasoning, where it is necessary to couple different contexts
appearing in a formula through a relation.
Supporting this approach within the context of \logic{} is a matter that
is currently under study.

Two other systems that are in existence that support reasoning about
LF specifications are Twelf~\cite{pfenning99cade} and
Beluga~\cite{pientka10ijcar}.
Both systems are based on the idea of describing functions whose
inputs and outputs are determined by LF typing judgements: by
exhibiting a function of this type that an external checker determines
to be total, one obtains a ``proof'' of a theorem that has a
$\forall\exists$ structure.
In Twelf, all theorems include an implicit universal quantification
over a context variable at the outermost level.
The domain of this quantifier is determined by a ``world'' description
that is similar to our context schemas and all the typing judgements
in the theorem are assumed to be relativized to the context that
instantiates the quantifier.
In contrast, Beluga allows contexts to be indicated explicitly with the
typing judgements.
Twelf allows for transportation of theorems between different world
descriptions when these are related by a property called \emph{world
subsumption}~\cite{harper07jfp} that, intuitively, encodes the idea
that all the contexts corresponding to one world description are
covered, relative to the theorem to be proved, by those corresponding
to the other.
The notion of context equivalence used is, in a sense, a
bi-directional version of the context expression subsumption that we
have described in Section~\ref{ssec:cssubsumption}.
However, by exploiting the particular way in which theorems are stated
in Twelf, context equivalence can be based directly on subordination
by a type and the context minimization operation described in
Section~\ref{ssec:subord}.
World subsumption can also use subordination directly, but how exactly
it is to be determined is left unspecified in the development
in~\cite{harper07jfp}.
Beluga, on the other hand, appears to support a functionality similar
to the one we have developed here.
However, we have not been able to find a formal development
of the idea and what is described in \cite{felty15jar} (see item 2 on
page 322) seems to lead to unsound reasoning: when trying to transport a
theorem of the form $\fctx{\Gamma}{\mathcal{C}}{F}$ into one of the
form $\fctx{\Gamma}{\mathcal{C}'}{F}$, it is necessary to check that
every block schema of $\mathcal{C}'$ is covered (in a suitable sense)
by one of $\mathcal{C}$, rather than the other way around as is
indicated in \cite{felty15jar}.\footnote{Experiments with version 1.1.2 of the
Beluga system indicate that what is implemented in the system is indeed the
unsound rule.}

In summary, the ideas we have investigated here have been considered
previously by others. However, our work differs from that manifest in
prior systems in that it has been carried out in the framework of a
logic that allows for explicit quantification over contexts and
that supports a more flexible statement of relationships between LF
typing judgements.
Moreover, we have developed a fully automated approach to
detecting context schema subsumption and have utilized the setting of
the logic to show the soundness of the transportation process that we
have based on this notion.

\begin{acks}
  This paper has benefited from the careful reading and the helpful
  comments of its reviewers.
\end{acks}



\begin{thebibliography}{19}

%%% the Web address.

\ifx \showCODEN    \undefined \def \showCODEN     #1{\unskip}     \fi
\ifx \showISBNx    \undefined \def \showISBNx     #1{\unskip}     \fi
\ifx \showISBNxiii \undefined \def \showISBNxiii  #1{\unskip}     \fi
\ifx \showISSN     \undefined \def \showISSN      #1{\unskip}     \fi
\ifx \showLCCN     \undefined \def \showLCCN      #1{\unskip}     \fi
\ifx \shownote     \undefined \def \shownote      #1{#1}          \fi
\ifx \showarticletitle \undefined \def \showarticletitle #1{#1}   \fi
\ifx \showURL      \undefined \def \showURL       {\relax}        \fi
% invisible to TeX
\providecommand\bibfield[2]{#2}
\providecommand\bibinfo[2]{#2}
\providecommand\natexlab[1]{#1}
\providecommand\showeprint[2][]{arXiv:#2}

\bibitem[Ashley-Rollman et~al\mbox{.}(2005)]%
        {ashley-rollman05poplmark}
\bibfield{author}{\bibinfo{person}{Michael Ashley-Rollman},
  \bibinfo{person}{Karl Crary}, {and} \bibinfo{person}{Robert Harper}.}
  \bibinfo{year}{2005}\natexlab{}.
\newblock \bibinfo{title}{A solution in {T}welf to the {\sc PoplMark}
  {C}hallenge problems 1 and 2}.
\newblock
\newblock
\shownote{Available from the URL
  \href{https://www.seas.upenn.edu/~plclub/poplmark/cmu.html}{https://www.seas.upenn.edu/~plclub/poplmark/cmu.html}}.


\bibitem[Aydemir et~al\mbox{.}(2005)]%
        {aydemir05tphols}
\bibfield{author}{\bibinfo{person}{Brian~E. Aydemir}, \bibinfo{person}{Aaron
  Bohannon}, \bibinfo{person}{Matthew Fairbairn}, \bibinfo{person}{J.~Nathan
  Foster}, \bibinfo{person}{Benjamin~C. Pierce}, \bibinfo{person}{Peter
  Sewell}, \bibinfo{person}{Dimitrios Vytiniotis}, \bibinfo{person}{Geoffrey
  Washburn}, \bibinfo{person}{Stephanie Weirich}, {and} \bibinfo{person}{Steve
  Zdancewic}.} \bibinfo{year}{2005}\natexlab{}.
\newblock \showarticletitle{Mechanized Metatheory for the Masses: The {\sc
  PoplMark} Challenge}. In \bibinfo{booktitle}{\emph{Theorem Proving in Higher
  Order Logics: 18th International Conference}} \emph{(\bibinfo{series}{Lecture
  Notes in Computer Science}, \bibinfo{number}{3603})}.
  \bibinfo{publisher}{Springer}, \bibinfo{pages}{50--65}.
\newblock
\href{https://doi.org/10.1007/11541868_4}{doi:\nolinkurl{10.1007/11541868_4}}


\bibitem[Felty et~al\mbox{.}(2015a)]%
        {felty15jar}
\bibfield{author}{\bibinfo{person}{Amy~P. Felty}, \bibinfo{person}{Alberto
  Momigliano}, {and} \bibinfo{person}{Brigitte Pientka}.}
  \bibinfo{year}{2015}\natexlab{a}.
\newblock \showarticletitle{The {Next} 700 {Challenge} {Problems} for
  {Reasoning} with {Higher}-{Order} {Abstract} {Syntax} {Representations}}.
\newblock \bibinfo{journal}{\emph{Journal of Automated Reasoning}}
  \bibinfo{volume}{55}, \bibinfo{number}{4} (\bibinfo{date}{Dec.}
  \bibinfo{year}{2015}), \bibinfo{pages}{307--372}.
\newblock
\showISSN{1573-0670}
\href{https://doi.org/10.1007/s10817-015-9327-3}{doi:\nolinkurl{10.1007/s10817-015-9327-3}}


\bibitem[Felty et~al\mbox{.}(2015b)]%
        {felty15arxiv}
\bibfield{author}{\bibinfo{person}{Amy~P. Felty}, \bibinfo{person}{Alberto
  Momigliano}, {and} \bibinfo{person}{Brigitte Pientka}.}
  \bibinfo{year}{2015}\natexlab{b}.
\newblock \showarticletitle{The Next 700 Challenge Problems for Reasoning with
  Higher-Order Abstract Syntax Representations: Part 1-A Common Infrastructure
  for Benchmarks}.
\newblock \bibinfo{journal}{\emph{CoRR}}  \bibinfo{volume}{abs/1503.06095}
  (\bibinfo{year}{2015}), \bibinfo{numpages}{42}~pages.
\newblock
\showeprint[arXiv]{1503.06095}
\urldef\tempurl%
\url{http://arxiv.org/abs/1503.06095}
\showURL{%
\tempurl}


\bibitem[Gacek et~al\mbox{.}(2011)]%
        {gacek11ic}
\bibfield{author}{\bibinfo{person}{Andrew Gacek}, \bibinfo{person}{Dale
  Miller}, {and} \bibinfo{person}{Gopalan Nadathur}.}
  \bibinfo{year}{2011}\natexlab{}.
\newblock \showarticletitle{Nominal Abstraction}.
\newblock \bibinfo{journal}{\emph{Information and Computation}}
  \bibinfo{volume}{209}, \bibinfo{number}{1} (\bibinfo{year}{2011}),
  \bibinfo{pages}{48--73}.
\newblock
\href{https://doi.org/10.1016/j.ic.2010.09.004}{doi:\nolinkurl{10.1016/j.ic.2010.09.004}}


\bibitem[Harper et~al\mbox{.}(1993)]%
        {harper93jacm}
\bibfield{author}{\bibinfo{person}{Robert Harper}, \bibinfo{person}{Furio
  Honsell}, {and} \bibinfo{person}{Gordon Plotkin}.}
  \bibinfo{year}{1993}\natexlab{}.
\newblock \showarticletitle{A Framework for Defining Logics}.
\newblock \bibinfo{journal}{\emph{J. ACM}} \bibinfo{volume}{40},
  \bibinfo{number}{1} (\bibinfo{year}{1993}), \bibinfo{pages}{143--184}.
\newblock
\href{https://doi.org/10.1145/138027.138060}{doi:\nolinkurl{10.1145/138027.138060}}


\bibitem[Harper and Licata(2007)]%
        {harper07jfp}
\bibfield{author}{\bibinfo{person}{Robert Harper} {and}
  \bibinfo{person}{Daniel~R. Licata}.} \bibinfo{year}{2007}\natexlab{}.
\newblock \showarticletitle{Mechanizing Metatheory in a Logical Framework}.
\newblock \bibinfo{journal}{\emph{Journal of Functional Programming}}
  \bibinfo{volume}{17}, \bibinfo{number}{4--5} (\bibinfo{date}{July}
  \bibinfo{year}{2007}), \bibinfo{pages}{613--673}.
\newblock
\href{https://doi.org/10.1017/S0956796807006430}{doi:\nolinkurl{10.1017/S0956796807006430}}


\bibitem[Nadathur and Southern(2021)]%
        {nadathur21arxiv}
\bibfield{author}{\bibinfo{person}{Gopalan Nadathur} {and}
  \bibinfo{person}{Mary Southern}.} \bibinfo{year}{2021}\natexlab{}.
\newblock \bibinfo{title}{A Logic for Reasoning About {LF} Specifications}.
  (\bibinfo{date}{June} \bibinfo{year}{2021}).
\newblock
\newblock
\shownote{Available from \url{http://arxiv.org/abs/2107.00111}.}.


\bibitem[Nadathur and Southern(2022)]%
        {nadathur22ppdp}
\bibfield{author}{\bibinfo{person}{Gopalan Nadathur} {and}
  \bibinfo{person}{Mary Southern}.} \bibinfo{year}{2022}\natexlab{}.
\newblock \showarticletitle{A Logic for Formalizing Properties of {LF}
  Specifications}. In \bibinfo{booktitle}{\emph{Proceedings of the 24th
  International Symposium on Principles and Practice of Declarative
  Programming}} (Tbilisi, Georgia) \emph{(\bibinfo{series}{PPDP '22})}.
  \bibinfo{publisher}{Association for Computing Machinery},
  \bibinfo{address}{New York, NY, USA}, Article \bibinfo{articleno}{6},
  \bibinfo{numpages}{13}~pages.
\newblock
\showISBNx{9781450397032}
\href{https://doi.org/10.1145/3551357.3551377}{doi:\nolinkurl{10.1145/3551357.3551377}}


\bibitem[Pfenning and Elliott(1988)]%
        {pfenning88pldi}
\bibfield{author}{\bibinfo{person}{Frank Pfenning} {and} \bibinfo{person}{Conal
  Elliott}.} \bibinfo{year}{1988}\natexlab{}.
\newblock \showarticletitle{Higher-Order Abstract Syntax}. In
  \bibinfo{booktitle}{\emph{Proceedings of the {ACM}-{SIGPLAN} Conference on
  Programming Language Design and Implementation}}. \bibinfo{publisher}{ACM
  Press}, \bibinfo{pages}{199--208}.
\newblock
\href{https://doi.org/10.1145/960116.54010}{doi:\nolinkurl{10.1145/960116.54010}}


\bibitem[Pfenning and Sch{\"u}rmann(1999)]%
        {pfenning99cade}
\bibfield{author}{\bibinfo{person}{Frank Pfenning} {and}
  \bibinfo{person}{Carsten Sch{\"u}rmann}.} \bibinfo{year}{1999}\natexlab{}.
\newblock \showarticletitle{System Description: Twelf --- {A} Meta-Logical
  Framework for Deductive Systems}. In \bibinfo{booktitle}{\emph{16th Conf.\ on
  Automated Deduction (CADE)}} \emph{(\bibinfo{series}{Lecture Notes in
  Artificial Intelligence}, \bibinfo{number}{1632})},
  \bibfield{editor}{\bibinfo{person}{H.~Ganzinger}} (Ed.).
  \bibinfo{publisher}{Springer}, \bibinfo{address}{Trento},
  \bibinfo{pages}{202--206}.
\newblock
\href{https://doi.org/10.1007/3-540-48660-7_14}{doi:\nolinkurl{10.1007/3-540-48660-7_14}}


\bibitem[Pfenning and Sch{\"u}rmann(2002)]%
        {Pfenning02guide}
\bibfield{author}{\bibinfo{person}{Frank Pfenning} {and}
  \bibinfo{person}{Carsten Sch{\"u}rmann}.} \bibinfo{year}{2002}\natexlab{}.
\newblock \bibinfo{booktitle}{\emph{Twelf User's Guide}}.
\newblock
\newblock
\shownote{Available from \url{http://www.cs.cmu.edu/~twelf/guide-1-4}}.


\bibitem[Pfenning et~al\mbox{.}({[n.\,d.]})]%
        {twelf.website}
\bibfield{author}{\bibinfo{person}{Frank Pfenning}, \bibinfo{person}{Carsten
  Sch\"{u}rmann}, \bibinfo{person}{Brigitte Pientka}, \bibinfo{person}{Roberto
  Virga}, {and} \bibinfo{person}{Kevin Watkins}.}
  \bibinfo{year}{[n.\,d.]}\natexlab{}.
\newblock \bibinfo{title}{The {T}welf Project}.
\newblock \bibinfo{howpublished}{\url{http://twelf.org/}}.
\newblock


\bibitem[Pientka and Dunfield(2010)]%
        {pientka10ijcar}
\bibfield{author}{\bibinfo{person}{Brigitte Pientka} {and}
  \bibinfo{person}{Joshua Dunfield}.} \bibinfo{year}{2010}\natexlab{}.
\newblock \showarticletitle{Beluga: {A} Framework for Programming and Reasoning
  with Deductive Systems (System Description)}. In
  \bibinfo{booktitle}{\emph{Fifth International Joint Conference on Automated
  Reasoning}} \emph{(\bibinfo{series}{Lecture Notes in Computer Science},
  \bibinfo{number}{6173})}, \bibfield{editor}{\bibinfo{person}{J.~Giesl} {and}
  \bibinfo{person}{R.~H{\"a}hnle}} (Eds.). \bibinfo{publisher}{Springer},
  \bibinfo{pages}{15--21}.
\newblock
\href{https://doi.org/10.1007/978-3-642-14203-1_2}{doi:\nolinkurl{10.1007/978-3-642-14203-1_2}}


\bibitem[Sch{\"{u}}rmann(2000)]%
        {schurmann00phd}
\bibfield{author}{\bibinfo{person}{Carsten Sch{\"{u}}rmann}.}
  \bibinfo{year}{2000}\natexlab{}.
\newblock \emph{\bibinfo{title}{Automating the Meta Theory of Deductive
  Systems}}.
\newblock \bibinfo{thesistype}{Ph.\,D. Dissertation}. \bibinfo{school}{Carnegie
  Mellon University}.
\newblock


\bibitem[Southern and Nadathur(2021)]%
        {southern21lfmtp}
\bibfield{author}{\bibinfo{person}{Mary Southern} {and}
  \bibinfo{person}{Gopalan Nadathur}.} \bibinfo{year}{2021}\natexlab{}.
\newblock \showarticletitle{{A}delfa: A System for Reasoning about {LF}
  Specifications}. In \bibinfo{booktitle}{\emph{Proceedings of the Sixteenth
  International Workshop on Logical Frameworks and Meta-Languages: Theory and
  Practice (LFMTP'21)}} \emph{(\bibinfo{series}{Electronic Proceedings in
  Theoretical Computer Science}, \bibinfo{number}{337})}.
  \bibinfo{publisher}{Open Publishing Association}, \bibinfo{pages}{104--120}.
\newblock
\href{https://doi.org/10.4204/eptcs.337.8}{doi:\nolinkurl{10.4204/eptcs.337.8}}


\bibitem[Tiu(2006)]%
        {tiu06lfmtp}
\bibfield{author}{\bibinfo{person}{Alwen Tiu}.}
  \bibinfo{year}{2006}\natexlab{}.
\newblock \showarticletitle{A Logic for Reasoning about Generic Judgments}. In
  \bibinfo{booktitle}{\emph{Proceedings of the First International Workshop on
  Logical Frameworks and Meta-Languages: Theory and Practice (LFMTP'06)}}
  \emph{(\bibinfo{series}{Electronic Notes in Theoretical Computer Science},
  Vol.~\bibinfo{volume}{174})},
  \bibfield{editor}{\bibinfo{person}{A.~Momigliano} {and}
  \bibinfo{person}{B.~Pientka}} (Eds.). \bibinfo{publisher}{Elsevier},
  \bibinfo{pages}{3--18}.
\newblock
Issue 5.
\href{https://doi.org/10.1016/j.entcs.2007.01.016}{doi:\nolinkurl{10.1016/j.entcs.2007.01.016}}


\bibitem[Virga(1999)]%
        {virga99phd}
\bibfield{author}{\bibinfo{person}{Roberto Virga}.}
  \bibinfo{year}{1999}\natexlab{}.
\newblock \emph{\bibinfo{title}{Higher-order Rewriting with Dependent Types}}.
\newblock \bibinfo{thesistype}{Ph.\,D. Dissertation}. \bibinfo{school}{Carnegie
  Mellon University}.
\newblock


\bibitem[Watkins et~al\mbox{.}(2003)]%
        {watkins03tr}
\bibfield{author}{\bibinfo{person}{Kevin Watkins}, \bibinfo{person}{Iliano
  Cervesato}, \bibinfo{person}{Frank Pfenning}, {and} \bibinfo{person}{David
  Walker}.} \bibinfo{year}{2003}\natexlab{}.
\newblock \bibinfo{booktitle}{\emph{A concurrent logical framework {I}:
  Judgments and properties}}.
\newblock \bibinfo{type}{{T}echnical {R}eport} CMU-CS-02-101.
  \bibinfo{institution}{Carnegie Mellon University}.
\newblock
\newblock
\shownote{Revised, May 2003}.


\end{thebibliography}
\end{document}